\documentclass[a4paper,12pt,fleqn]{article}
\usepackage{amsmath,amssymb}
\usepackage{amsthm}
\usepackage{graphicx}
\usepackage{float}
\usepackage[np]{numprint}
\usepackage{calc}    
\usepackage{tikz,tkz-tab}
\usepackage{diagbox}
\usepackage{multirow}
\usepackage{braket}
\usepackage{varioref}
\usepackage{hyperref}
\usepackage{lipsum}
\usepackage[left=3cm,right=3cm,top=2.5cm,bottom=4cm]{geometry}
\usepackage{multirow}
\usetikzlibrary{decorations.pathreplacing}

\newcommand\swap{\mathtt{SWAP}}
\newcommand\cnot{\mathtt{CNOT}}
\newcommand\PauliX{Pauli-$\mathtt X$ }
\newcommand\PauliY{Pauli-$\mathtt Y$ }
\newcommand\PauliZ{Pauli-$\mathtt Z$ }

\renewcommand\phi{\varphi}
\renewcommand\epsilon{\varepsilon}
\renewcommand\geq{\geqslant}
\renewcommand\leq{\leqslant}

\newcommand\Z{\mathbb{Z}}
\newcommand\C{\mathbb{C}}
\newcommand\F{\mathbb{F}_{2}}
\newcommand\ee{\mathrm e}
\newcommand\ii{\mathrm i}
\newcommand\HS{\mathcal{H}} 

\newcommand\symg[1][n]{\mathfrak{S}_{#1}}  
\newcommand\swapg[1][n]{{\left\langle\swap\right\rangle}_{#1}} 
\newcommand\cnotg[1][n]{{\left\langle\cnot\right\rangle}_{#1}} 
\newcommand\GL[1][n]{\mathrm{GL}_{#1}(\mathbb{F}_2)} 
\newcommand\UG[1][2n]{\mathcal{U}_{2^n}} 

\newcommand\param{\mathrm{PMAX}}
\newtheorem{example}{Example}

\newtheorem{prop}[example]{Proposition}
\newtheorem{cor}[example]{Corollary}
\newtheorem{lem}[example]{Lemma}

\newtheorem{rem}[example]{Remark}
\newtheorem{conj}[example]{Conjecture}

\begin{document}
\setlength\parindent{0mm}

\overfullrule=0mm
\floatstyle{boxed} 
\restylefloat{figure}

\title{Quantum circuits generating four-qubit  maximally entangled states}

\author{Marc Bataille \\ marc.bataille1@univ-rouen.fr \\
  \\ LITIS laboratory, Universit\'e Rouen-Normandie \thanks{685 Avenue de l'Universit\'e, 76800 Saint-\'Etienne-du-Rouvray. France.}}

\date{}

\maketitle

\begin{abstract}
  We describe quantum circuits generating four-qubit maximally entangled states, the amount of entanglement being quantified by using the absolute value of the Cayley hyperdeterminant as an entanglement monotone. More precisely we show that this type of four-qubit entangled states can be obtained by the action of a family of CNOT circuits on some special states of the LU orbit of the state $\ket{0000}$.
\end{abstract}

\section{Introduction}

The original idea of using the hyperdeterminant  to classify multipartite entanglement goes back to Miyake \cite{2002MW, 2003Miyake}.
The hyperdeterminant (in the sense of Gelfand \textit{et al.} \cite{1992GKL}) is a generalization of the  determinant to higher dimensions.  
Let $\ket{\psi}=\sum_{i_0, i_1,\dots,i_{n-1}\in\{0,1\}}a_{i_0 i_1\dots i_{n-1}}\ket{i_0i_1\dots i_{n-1}}$ be the state vector of an $n$-qubit system in the Hilbert space $\HS^{\otimes n}=(\C^2)^{\otimes n}$, then the hyperdeterminant of the format $2^n$, denoted in this paper by $\Delta_n$, is an homogenous multivariate polynomial in the $2^n$ variables $a_{i_0 i_1\dots i_{n-1}}$, with coefficients in $\Z$. It is invariant (up to a sign) by permutation of the qubits and also invariant by the action of the group $SLOCC$,  the group of stochastic local operations assisted by classical communication, assimilated to the cartesian product $SL(2,\C)^n$.
According to Miyake \cite{2002MW, 2003Miyake}, the more generic entanglement holds only for the states on which the hyperdeterminant does not vanish and the absolute value of $\Delta_n$ quantifies the amount of generic entanglement.
\medskip

In this article, we focus on a four-qubit quantum system. In this case, the hyperdeterminant $\Delta_4$ is of degree 24 and an expression of $\Delta_4$ in terms of fundamental SLOCC invariant polynomials of lower degree was given by Luque and Thibon \cite{2003LT}. 
Following Miyake, we consider as Gour and Wallach \cite{2012GW}, that a four-qubit state with the highest amount of generic entanglement can be defined as a state maximizing the absolute value of $\Delta_4$. In the rest of the paper, we refer to this type of state as a \textit{maximum hyperdeterminant state} (sometimes abbreviated as MHS).
In a paper from 2012 \cite{2012GW}, Gour and Wallach conjectured that the state $\ket{L}$ (see Fig. \ref{delta_max}) is the unique maximum hyperdeterminant state, up to a local unitary operation. This conjecture was proved in 2013 by Chen and Djokovic \cite{2013CD} and the maximal value of $|\Delta_4|$ is $\frac{1}{2^83^9}=\frac{1}{\np{5038848}}\simeq 1.98\times10^{-7}$. Moreover, the state $\ket{L}$ has also the property to be the only state (up to local unitary operations) maximizing the average Tsallis $\alpha$-entropy of entanglement, for all $\alpha > 2$ \cite{2010GW}.
Let us also mention  two other maximum hyperdeterminant states, which have the property of having real coordinates (see Fig. \ref{delta_max}): $\ket{\Phi_5}$ (reported by Osterloh and Siewert in \cite{2006OS} and by Alsina in his PhD thesis \cite{2017Alsina}) and $\ket{M_{2222}}$ (reported by Hamza Jaffali in his PhD thesis, unpublished) .
\medskip


In Quantum Information and Computation, entangled states, and in particular maximally entangled states, play the role of an important physical resource (see \textit{e.g.} the introduction of \cite{2013CD}). Despite of that, to our knowledge,  there is no proposal in the academic literature for quantum circuits capable of producing the state $\ket{L}$ or any other MHS. The goal of this work is merely to fill this gap  by describing a family of quantum circuits that enable the generation of maximum hyperdeterminant states.
We show that a MHS can be obtained by the action a certain type of CNOT gate circuits on a fully factorized state, namely a state of the LU orbit of $\ket{0000}$. As a consequence of this result, one can construct quantum circuits of relatively small depth generating the three states $\ket{L}$, $\ket{\Phi_5}$ and $\ket{M_{2222}}$.
\medskip

The paper is structured as follows. Section \ref{background} is a reminder on quantum circuits of $\cnot$ gates and $\swap$ gates, where we introduce most of our notations and give some useful conjugation rules between these gates. In Section \ref{method}, we present the methodology and algorithms used in our numerical approach to find circuits generating maximum hyperdeterminant states. The two next sections (\ref{cnot-max} and \ref{cosets}) are dedicated to the description of these circuits. Finally, in Section \ref{generation}, we propose three simple quantum circuits generating the states $\ket{L}$, $\ket{\Phi_5}$ and $M_{2222}$,
as well as an implementation of a circuit generating the state $\ket{L}$ into a quantum computer provided by the IBM quantum experience at \url{https://quantum-computing.ibm.com/}.
\medskip

This article goes along with a Python module than can be downloaded at \url{https://github.com/marcbataille/maximum-hyperdeterminant-states}.
The module provides an implementation of the different algorithms, quantum gates and quantum states used in this work. As the proof of some assertions (mostly numerical equalities) consists only of basic linear algebra and calculus, we chose to refer the reader to the corresponding function of the module that does the job.

       \begin{figure}[h]
     \begin{equation}
        \ket{L}=\frac{1}{\sqrt3}(\ket{u_0}+\omega\ket{u_1}+\omega^*\ket{u_2})\label{L-def}
      \end{equation}
      $\text{where : }\omega=\ee^{\frac{\ii\pi}{3}}$

      $\phantom{where : }\ket{u_0}=\frac{1}{2}(\ket{0000} + \ket{0011} + \ket{1100} + \ket{1111})=\ket{\Phi^+}\ket{\Phi^+}$
      
      $\phantom{where : }\ket{u_1}=\frac{1}{2}(\ket{0000} - \ket{0011} - \ket{1100} + \ket{1111})=\ket{\Phi^-}\ket{\Phi^-}$
      
      $\phantom{where : }\ket{u_2}=\frac{1}{2}(\ket{0101} + \ket{0110} + \ket{1001} + \ket{1010})=\ket{\Psi^+}\ket{\Psi^+}$

      $\phantom{where : }\ket{\Phi^{\pm}}=\frac{1}{\sqrt 2}(\ket{00}\pm \ket{11}),\quad\ket{\Psi^{\pm}}=\frac{1}{\sqrt 2}(\ket{01}\pm \ket{10})$
      
      \begin{equation}
        \ket{\Phi_5} = \frac{1}{\sqrt6}(\ket{0001} + \ket{0010} + \ket{0100} + \ket{1000} + \sqrt{2}\ket{1111})
      \end{equation}

      \begin{equation}
        \ket{M_{2222}} = \frac{1}{\sqrt8}\ket{v_1}+\frac{\sqrt6}{4}\ket{v_2}+\frac{1}{\sqrt2}\ket{v_3}
      \end{equation}
            $\text{where: } \ket{v_1}=\frac{1}{\sqrt{6}}(\ket{0000} + \ket{0101} - \ket{0110} - \ket{1001} + \ket{1010} + \ket{1111})$
      
      $\phantom{where: }\ket{v_2}=\frac{1}{\sqrt{2}}(\ket{0011} + \ket{1100})$
      
      $\phantom{where: }\ket{v_3}=\frac{1}{\sqrt{8}}(-\ket{0001} + \ket{0010} - \ket{0100} + \ket{0111} + \ket{1000} - \ket{1011} + \ket{1101} - \ket{1110})$
      \caption{4-qubits states for which $|\Delta_4|$ is maximal. \label{delta_max}}
    \end{figure}

\section{Quantum circuits of $\cnot$ and $\swap$ gates\label{background}}
In this section we introduce the main notations and conventions of the paper and we recall the definition of some classical quantum gates (Table \ref{classic}) as well as some properties of the $\cnot$ gates and $\swap$ gates often used in the rest of the article.\medskip

Let $n\geq 1$ be the number of qubits of the considered quantum register. We label each qubit from 0 to $n-1$, thus following the usual convention. For coherence we also number the lines and columns of a $n\times n$ matrix from 0 to $n-1$ and we consider that a permutation in the symmetric group $\symg$ is a bijection of $\{0,\dots,n-1\}$. 

\begin{table}
  \begin{center}
\begin{tabular}{|c|c|c|}\hline
  Name&Symbol&Matrix\\\hline\hline
  \PauliX& $X$&$\begin{bmatrix}0&1\\1&0\end{bmatrix}$\\
  \PauliY&$Y$&$\begin{bmatrix}0&-\ii\\\ii&0\end{bmatrix}$\\
  \PauliZ&$Z$&$\begin{bmatrix}1&0\\0&-1\end{bmatrix}$\\
  Rotation around the $\hat x$ axis&$R_x(\theta)$&$\ee^{-\ii\theta X/2}=
                                                   \begin{bmatrix}\cos\frac{\theta}{2}&-\ii\sin\frac{\theta}{2}
  \\-\ii\sin\frac{\theta}{2}&\cos\frac{\theta}{2}\end{bmatrix}$\\
  Rotation around the $\hat y$ axis&$  R_y(\theta)$&$\ee^{-\ii\theta Y/2}=
                                                     \begin{bmatrix}\cos\frac{\theta}{2}&-\sin\frac{\theta}{2}\\\sin\frac{\theta}{2}&\cos\frac{\theta}{2}\end{bmatrix}$\\
  Rotation around the $\hat z$ axis&$R_z(\theta)$&$\ee^{-\ii\theta Z/2}=
                                                   \begin{bmatrix}\ee^{-\ii\theta/2}&0\\0&\ee^{\ii\theta/2}\end{bmatrix}$\\
  Phase&$P$&$\begin{bmatrix}1&0\\0&\ii\end{bmatrix}$\\
  T-gate&$T$&$\begin{bmatrix}1&0\\0&\ee^{\ii\pi/4}\end{bmatrix}$\\
  Hadamard&$H$&$\frac{\sqrt 2}{2}\begin{bmatrix}1&1\\1&-1\end{bmatrix}$\\\hline
\end{tabular}
\end{center}
\caption{Classical single qubit unitary gates\label{classic}}
\end{table}

\medskip

If two normalized vectors $\ket{\psi}$ and $\ket{\psi'}$  of the Hilbert space $\HS^{\otimes n}$ are equal up to a global phase, then they represent physically the same state and we write
$\ket{\psi}\simeq\ket{\psi'}$. In the same way, we write $U\simeq U'$ for two unitary operators which are equal up to a global phase.
In the design of quantum circuits, we use the following correspondences between the classical gates :
\begin{align}
  &R_z(\pi)\simeq Z,\ 
    R_z(\pi/2)\simeq P,\  
    R_z(-\pi/2)\simeq P^{\dag},\  
  R_z(\pi/4)\simeq T\label{Rz-classical}\\
  &R_y(\pi)\simeq Y,\ 
  R_y(\pi/2)= H Z=X H,\ 
  R_y(-\pi/2)=ZH=HX\label{Ry-classical}
  \end{align}
When we apply locally a single-qubit gate $U$ to the qubit $i$ of a $n$-qubit register, the corresponding action on the $n$-qubit system is that of the unitary operator
\begin{equation}
  U_i=\underbrace{I\otimes\dots\otimes I}_{i\ \mathrm{times}} \otimes U \otimes\underbrace{I \otimes \dots \otimes I}_{n-i-1 \text{ times}}=I^{\otimes i}\otimes U \otimes I^{\otimes n-i-1},\label{single}
  \end{equation}
where $\otimes$ is the Kronecker product of matrices and $I$ the identity matrix in dimension 2.
As an example, if $n=4$, $H_1= I\otimes H\otimes I\otimes I$ and $H_0H_3=H\otimes I\otimes I\otimes H$.
We also use  vectors of $\F^n$ as labels to indicate the set of qubits on  which the single-qubit unitary $U$ is applied. Let $v=[v_0,\dots,v_{n-1}]^t$ be a (column) vector of $\F^n$, we denote by $U_{v}$ the product $\prod_iU_i^{v_i}$.
\medskip

A $\cnot$ gate with target on qubit $i$ and control  on qubit $j$ is denoted by $X_{[ij]}$ (not to be confused with $X_i$ which denotes a \PauliX gate applied on qubit $i$).
The group generated by the $\cnot$ gates acting on an $n$-qubit quantum system is denoted by $\cnotg$. Let us denote by $\GL$ the general linear group over $\F$ in dimension $n$. A transvection matrix $[ij]$ ($i,j=1\dots n-1$ and $i\neq j$), is the matrix of $\GL$ defined by $[ij]=I_n+E_{ij}$,  where $I_n$ is the identity in dimension $n$ and $E_{ij}$ is the matrix with all entries equal to zero but the entry $(i,j)$ that is equal to 1. We recall that the transvection matrices generate the group $\GL$ and that multiplying a matrix $M$ to the left by a transvection matrix $[ij]$ is equivalent to adding the row $j$ to the row $i$ of $M$. From these facts, one can deduce that the group $\cnotg$ is isomorphic to $\GL$, a possible isomorphism associating, to any gate $X_{[ij]}$, the transvection matrix $[ij]$ (see \cite{2020B} for more details).
The order of $\cnotg$ is therefore equal to the order of $\GL$ :
\begin{equation}
|\cnotg|=2^{\frac{n(n-1)}{2}}\prod_{i=1}^n(2^i-1).\label{order-cnot}
  \end{equation}
Let $A$ be any  matrix in $\GL$, we denote by $X_A$ the element of $\cnotg$ associated to $A$, \textit{i.e.} $X_A$ is the product of any sequence of $\cnot$ gates $X_{[i_1,j_i]},\dots, X_{[i_p,j_p]}$ such that $A$ can be decomposed in the product of the transvection matrices $[i_1,j_i],\dots,[i_p,j_p]$.
\medskip

The $\swap$ gate that exchanges qubits $i$ and $j$ is denoted by $S_{(ij)}$.
Let $\sigma$ be a permutation of the symmetric group $\symg$. We also denote by $\sigma$ the permutation matrix associated to the permutation $\sigma$. This matrix is defined as the matrix $A=(a_{ij})$ in $\GL$ such that
$a_{ij}=1$ if and only if $i=\sigma(j)$. We recall that multiplying a matrix $M$ to the left by $\sigma$ is equivalent to applying the permutation $\sigma$ to the rows of $M$. In this case, each row $R_i$ is replaced by the row $R_{\sigma^{-1}(i)}$. The group of permutation matrices is a subgroup of $\GL$ which is isomorphic to the group generated by the $\swap$ gates acting on $n$ qubits : to each $\swap$ gate $S_{(ij)}$ corresponds the transposition matrix $(ij)$. We denote by $S_{\sigma}$ the product of any sequence of $\swap$ gates $S_{(i_1j_1)},\dots, S_{(i_p,j_p)}$ such that $\sigma=(i_1j_1)\dots (i_pj_p)$.
Let $\tau$ be a transposition of $\symg$, it is easy to check that
  $S_{\tau}X_{[ij]}S_{\tau}=X_{[\tau(i)\tau(j)]}$, hence by induction 
\begin{equation}
  S_{\sigma}X_{[ij]}S_{\sigma}^{-1}=X_{[\sigma(i)\sigma(j)]},\label{Xij-conj-sigma}
\end{equation}
for any permutation $\sigma$.
Let $U$ be a single-qubit unitary matrix and $U_i$ the unitary corresponding to the action of $U$ on qubit $i$ (Identity \eqref{single}), one has for any permutation $\sigma$
and any vector $v$ in $\F^{n}$ :
\begin{align}
&S_{\sigma}U_iS_{\sigma}^{-1}=U_{\sigma(i)}\label{conj-sigma-U}\\
&S_{\sigma}U_vS_{\sigma}^{-1}=U_{\sigma v}\label{conj-sigma-v}
\end{align}
\medskip

 The Pauli group for $n$ qubits is the group generated by the Pauli gates $X_i,Y_i$ and $Z_i$ ($0\leq i \leq n-1$). Since $Y=\ii XZ$ and $XZ=-ZX$, any element of this group can be written uniquely in the form
  \begin{equation}
    \ii^{\lambda}X_uZ_v,\label{Pauli-group}
\end{equation} where $u$ and $v$ are two vectors of the space $\F^n$ and $\lambda\in\{0,1,2,3\}$.\medskip 

It is not difficult to prove the following conjugation rules of a Pauli gate by a $\cnot$ gate :
$X_{[ij]}Z_iX_{[ij]}=Z_iZ_j$, $X_{[ij]}Z_jX_{[ij]}=Z_j$, $X_{[ij]}X_iX_{[ij]}=X_i$ and $X_{[ij]}X_jX_{[ij]}=X_iX_j$. These rules can be generalized as
  \begin{equation}
    X_AX_uZ_vX_A^{-1}=X_{Au}Z_{A^{-t} v},\label{Pauli-conj}
  \end{equation}
  where $u$ and $v$ are vectors in $\F^n$, $A$ is a matrix in $\GL$ and $A^{-t}$ a shorthand for $\left(A^{-1}\right)^{t}$.
  
\section{Methodology used in the numerical exploration\label{method}}
We address the following problem : is it possible to generate a state maximizing $\Delta_4$ by applying a $\cnot$ gate circuit on a state of the LU orbit of $\ket{0000}$ ?
We use the classical Z-Y decomposition of a single qubit unitary operator in the form $\ee ^{\ii\phi}R_z(\alpha)R_y(\beta)R_z(\alpha')$ (see \textit{e.g.} \cite[Th. 4.1]{2011NC}).
Using this decomposition, any fully factorized unitary operator $U$ depends, up to a global phase, on the 12 real parameters of the matrix 

\begin{equation}
   \mathcal P=\begin{bmatrix}\alpha_0&\beta_0&\alpha_0'\\
     \alpha_1&\beta_1&\alpha_1'\\
     \alpha_2&\beta_2&\alpha_2'\\
     \alpha_3&\beta_3&\alpha_3'\end{bmatrix}.
\end{equation}

We define the unitary $U(\mathcal P)$ by 
\begin{align}
  \begin{split}
    U(\mathcal P) = R_z(\alpha_0)R_y(\beta_0)R_z(\alpha_0')&\otimes R_z(\alpha_1)R_y(\beta_1)R_z(\alpha_1')\\
    \otimes R_z(\alpha_2)&R_y(\beta_2)R_z(\alpha_2')\otimes R_z(\alpha_3)R_y(\beta_3)R_z(\alpha_3').
    \end{split}
  \end{align}
As a rotation around the $\hat z$ axis applied to $\ket{0}$ is just a change of phase, it is possible to write any state vector of the LU orbit of  $\ket{0000}$ (up to a global phase), by using only two parameters for each qubit. So, any state in the LU orbit of $\ket{0000}$ is equal (up to a global phase) to the state $\ket{\mathcal P}$ defined by 
\begin{equation}
\ket{\mathcal P}=R_z(\alpha_0)R_y(\beta_0)\otimes R_z(\alpha_1)R_y(\beta_1)\otimes R_z(\alpha_2)R_y(\beta_2)\otimes R_z(\alpha_3)R_y(\beta_3)\ket{0000}.
\end{equation}
Using the definition of the rotation matrices around the $\hat z$ and $\hat y$ axes, one has
\begin{equation}
\ket{\mathcal P}=(a_0\ket{0} + a_1\ket{1})\otimes(b_0\ket{0} + b_1\ket{1})\otimes(c_0\ket{0} + c_1\ket{1})\otimes (d_0\ket{0} + d_1\ket{1}),
  \end{equation}
  $\text{where : }(a_0,a_1)=(\ee^{-\ii\alpha_0/2}\cos\frac{\beta_0}{2}, \ee^{\ii\alpha_0/2}\sin\frac{\beta_0}{2})$,
  
  $\phantom{\text{where : }}(b_0,b_1)=(\ee^{-\ii\alpha_1/2}\cos\frac{\beta_1}{2},\ee^{\ii\alpha_1/2}\sin\frac{\beta_1}{2})$,
  
  $\phantom{\text{where : }}(c_0,c_1)=(\ee^{-\ii\alpha_2/2}\cos\frac{\beta_2}{2},\ee^{\ii\alpha_2/2}\sin\frac{\beta_2}{2})$,
  
  $\phantom{\text{where : }}(d_0,d_1)=(\ee^{-\ii\alpha_3/2}\cos\frac{\beta_3}{2},\ee^{\ii\alpha_3/2}\sin\frac{\beta_3}{2})$.
  \medskip
  
  Any state resulting from the action of an unitary operator in $\cnotg[4]$, on a state of the LU orbit of $\ket{0000}$ can be written (up to a global phase)
  in the form $X_A\ket{\mathcal P}$,
where $A$ is a matrix in $\GL[4]$ and $\mathcal P$ a matrix of parameters.\medskip

In order to determine the states of type $X_A\ket{\mathcal P}$ capable of maximizing $|\Delta_4|$, it is sufficient to consider the right cosets of the subgroup of $\cnotg[4]$  generated by the $\swap$ gates (group $\swapg[4]\simeq \symg[4])$, because $|\Delta_4|$ is invariant under permutation of the qubits,
\textit{i.e.} $|\Delta_4(X_A\ket{\mathcal P})|=|\Delta_4(X_{\sigma A}\ket{\mathcal P})|$ for any permutation matrix $\sigma$.
The order of the group $\cnotg[4]$ is 20160 (Identity \eqref{order-cnot}), so the number of right cosets of $\swapg[4]$ in $\cnotg[4]$  is $20160/24=840$. For each coset, we compute a representative of minimal length in the generators $X_{[ij]}$ (function \texttt{right\_cosets\_perm\_GL4} of the Python module).
\medskip

The computation of $\Delta_4$ for a given state is performed using the algorithm proposed by Luque and Thibon in \cite[Section IV]{2003LT} (function \texttt{hyper\_det} of the Python module). After eliminating all coset representatives $X_A$ such that $|\Delta_4(X_A\ket{\mathcal P})|$ vanishes for any $\mathcal P$, we obtain a list of 333 representatives (function \texttt{non\_zero\_HD\_strings} of the Python module).  For each of them, we use a random walk on the search space defined by the eight parameters of $\ket{\mathcal P}$ in order to maximize the value of $|\Delta_4(X_A\ket{\mathcal P})|$ (function \texttt{search\_max\_HD} of the Python module). We check that it is possible to reach the maximal value of $\frac{1}{2^83^9}$ for $|\Delta_4|$ (accuracy $10^{-22}$) for only 12 coset representatives. These cosets are described in Section \ref{cosets}. Finally, from the approximate values of $\mathcal P$ computed by the random walk heuristic, it is possible to deduce the exact values of $\mathcal P$ such that  $|\Delta_4(X_A\ket{\mathcal P})|=\frac{1}{2^83^9}$.

\section{A $\cnot$ circuit to reach the maximum of $|\Delta_4|$\label{cnot-max}}
Let $i,j,k,\ell$ be distinct integers in $\{0,1,2,3\}$. We define $M_{k}^{(i,j)}$, a product of $\cnot$ gates, and $A_{k}^{(i,j)}$, the bit matrix of $\GL[4]$ associated to $M_{k}^{(i,j)}$, as follows :
\begin{align}
  M_{k}^{(i,j)}&=X_{[ij]}X_{[jk]}X_{[ki]}X_{[i\ell]}X_{[\ell j]}\\
  A_{k}^{(i,j)}&=[ij][jk][ki][i\ell][\ell j]
\end{align}
In this section, we show how to reach the maximum of $|\Delta_4|$ using the operator
\begin{equation}
  M_{2}^{(0,1)}=X_{[01]}X_{[12]}X_{[20]}X_{[03]}X_{[31]}.\label{M01-2}
  \end{equation}
  The results are extended to any operator of type $M_{k}^{(i,j)}$ in the next section.
  Since $A_{2}^{(0,1)}=[01][12][20][03][31]=\begin{bmatrix}0&1&1&0\\1&0&1&1\\1&1&1&1\\0&1&0&1\end{bmatrix}$ 
and $A_{3}^{(0,1)}=[01][13][30][02][21]=\begin{bmatrix}0&1&0&1\\1&0&1&1\\0&1&1&0\\1&1&1&1\end{bmatrix}$, we remark that $A_{3}^{(0,1)}=(023)A_{2}^{(0,1)}$,
so
\begin{equation}
  M_{3}^{(0,1)}=S_{(023)}M_{2}^{(0,1)}\label{same-coset},
\end{equation}
which means that $M_{3}^{(0,1)}$ and $M_{2}^{(0,1)}$ represent the same coset. This coset is denoted by
$\overline{(0,1)}$.
\begin{prop}\label{generate-psi}
  Let $\mathcal P_{\mathrm{max}}$ and $\mathcal P_{\mathrm{max}}'$ be the two matrices of parameters defined by
  \begin{equation}
\mathcal P_{\mathrm{max}}=\begin{bmatrix}\pi/2&\pi/2&0\\
  \pi/2&\pi/2&0\\
  \pi/4&\cos^{-1}\frac{\sqrt 3}{3}&0\\
  \pi/4&\cos^{-1}\frac{\sqrt 3}{3}&0\end{bmatrix},
\qquad
\mathcal P_{\mathrm{max}}'=\begin{bmatrix}\pi/2&\pi/2&0\\
  \pi/2&\pi/2&0\\
  3\pi/4&\cos^{-1}\frac{\sqrt 3}{3}&0\\
  3\pi/4&\cos^{-1}\frac{\sqrt 3}{3}&0\end{bmatrix},\label{P-max}
\end{equation}
then the states
\begin{equation}
\ket{\psi_{\mathrm{max}}}=M_{2}^{(0,1)}\ket{\mathcal P_{\mathrm{max}}}\label{psi-max}
\end{equation}
and
\begin{equation}
  \ket{\psi_{\mathrm{max}}'}=M_{2}^{(0,1)}\ket{\mathcal P'_{\mathrm{max}}}\label{psi-max-prime}
\end{equation}
maximize the absolute value of the four-qubit hyperdeterminant. One has
\begin{equation}
  \Delta_4(\ket{\psi_{\mathrm{max}}})=\Delta_4(\ket{\psi_{\mathrm{max}}'})=-\frac{1}{2^{8}3^{9}},
  \end{equation}
\begin{equation}
  \ket{\psi_{\mathrm{max}}}=\frac{\sqrt3}{3}\ket{w_1}+\frac{3+\sqrt 3}{6}\ee^{\ii\frac{\pi}{4}}\ket{w_2}+\frac{3-\sqrt 3}{6}\ee^{\ii\frac{\pi}{4}}\ket{w_3}
\end{equation}
and
\begin{equation}
  \ket{\psi_{\mathrm{max}}'}=\frac{\sqrt3}{3}\ket{w_1}+\frac{3+\sqrt 3}{6}\ee^{-\ii\frac{\pi}{4}}\ket{w_2}+\frac{3-\sqrt 3}{6}\ee^{\ii\frac{3\pi}{4}}\ket{w_3},
\end{equation}
where \begin{align*}
            \ket{w_1}&=\frac{1}{\sqrt 8}(\ket{0001}+\ii\ket{0011}+\ket{0101}-\ii\ket{0111}+\ket{1000}+\ii\ket{1010}+\ket{1100}-\ii\ket{1110}),\\
            \ket{w_2}&=\frac{1}{2}(-\ket{0000}-\ii\ket{0110}-\ii\ket{1011} +\ket{1101}),\\
            \ket{w_3}&=\frac{1}{2}(\ket{0010}+\ii\ket{0100}-\ii\ket{1001} + \ket{1111}).
            \end{align*}
          \end{prop}
          
          \begin{proof}
            The different assertions can be checked using the function
            
            \texttt{check\_psi\_max\_is\_MHS} of the Python module.
\end{proof}

In our numerical search for matrices of parameters $\mathcal P$ such that $M_{2}^{(0,1)}\ket{\mathcal P}$ maximizes $\Delta_4$, it appears that all values of $\mathcal P$ computed by the random walk heuristic are related to $\mathcal P_{\mathrm{max}}$ or to $\mathcal P_{\mathrm{max}}'$ by simple operations. These operations are described by the following lemma and its corollary.
Numerical results suggest that these operations applied to the matrices $\mathcal P_{\mathrm{max}}$ or $\mathcal P_{\mathrm{max}}'$ are sufficient to describe all the possible matrices $\mathcal P$ such that $M_{2}^{(0,1)}\ket{\mathcal P}$ is a MHS (Conjecture \ref{ops-conj}).

\begin{lem}\label{P-ops-Pauli} Let $\mathcal P$ be a matrix of parameters and, for any $k$ in $\{0,1,2,3\}$, let us denote by :
  
  $\mathcal P_{\alpha_k+\pi}$  the matrix obtained from $\mathcal P$ by adding $\pi$ to the parameter $\alpha_k$,
  
  $\mathcal P_{-\beta_k}$, the matrix obtained from $\mathcal P$ by taking the opposite of $\beta_k$,

  $\mathcal P_{-\alpha_k,\  \beta_k+\pi}$, the matrix obtained from $\mathcal P$ by taking the opposite of $\alpha_k$ and adding  $\pi$ to $\beta_k$ .
  Then :
\begin{align}
  &\ket{\mathcal P_{\alpha_k+\pi}}=-\ii Z_k\ket{\mathcal P}\label{alpha+pi}\\
  &\ket{\mathcal P_{-\beta_k}}= Z_k\ket{\mathcal P}\label{-beta}\\
  &\ket{\mathcal P_{-\alpha_k,\ \beta_k+\pi}}=-\ii Y_k\ket{\mathcal P}\label{-alpha-beta+pi}
\end{align}
  \end{lem}

  \begin{proof}We prove only Identity \eqref{-alpha-beta+pi}, the proofs of Identities \eqref{alpha+pi} and \eqref{-beta} being similar. Without loss of generality, we suppose that the last column of the matrix $\mathcal P$ is null and $k=0$. On the one hand :
    
    $-\ii Y_0\ket{\mathcal P}=(-\ii Y(a_0\ket{0} + a_1\ket{1}))\otimes(b_0\ket{0} + b_1\ket{1})\otimes(c_0\ket{0} + c_1\ket{1})\otimes (d_0\ket{0} + d_1\ket{1})$,
    $\phantom{-\ii Y_0\ket{\mathcal P}}=(-a_1\ket{0} + a_0\ket{1})\otimes(b_0\ket{0} + b_1\ket{1})\otimes(c_0\ket{0} + c_1\ket{1})\otimes (d_0\ket{0} + d_1\ket{1})$,
    where $(a_0,a_1)=(\ee^{-\ii\alpha_0/2}\cos\frac{\beta_0}{2}, \ee^{\ii\alpha_0/2}\sin\frac{\beta_0}{2})$,

    On the other hand :

    $\ket{\mathcal P_{-\alpha_0,\ \beta_0+\pi}}=(a_0'\ket{0} + a_1'\ket{1})\otimes(b_0\ket{0} + b_1\ket{1})\otimes(c_0\ket{0} + c_1\ket{1})\otimes (d_0\ket{0} + d_1\ket{1})$,
    where $(a_0',a_1')=(\ee^{-\ii(-\alpha_0/2)}\cos\frac{\beta_0 + \pi}{2}, \ee^{-\ii\alpha_0/2}\sin\frac{\beta_0+\pi}{2})=(-a_1,a_0)$.
    
    Hence $-\ii Y_0\ket{\mathcal P}=\ket{\mathcal P_{-\alpha_0,\ \beta_0+\pi}}$.
  \end{proof}

\begin{cor}\label{P-ops-max}
  Let $A$ be a matrix in $\GL[4]$ and $\mathcal P$ a matrix of parameters.

  If $|\Delta_4|$ is maximal for $X_A\ket{\mathcal P}$, then $|\Delta_4|$ is also maximal for $X_A\ket{\mathcal P_{\alpha_k+\pi}}$,
   $X_A\ket{\mathcal P_{-\beta_k}}$ and $X_A\ket{\mathcal P_{-\alpha_k,\ \beta_k+\pi}}$, for any $k$ in $\{0,1,2,3\}$. 
 \end{cor}

 \begin{proof} Suppose that $|\Delta_4|$ is maximal for $X_A\ket{\mathcal P}$.
   Let $\mathcal P'\in\{\mathcal P_{\alpha_k+\pi}, \mathcal P_{-\beta_k}, \mathcal P_{-\alpha_k,\ \beta_k+\pi}\}$. From Lemma \ref{P-ops-Pauli} and Identity \eqref{Pauli-group}, there exists two vectors $u$ and $v$ in $\F^4$ such that $\ket{\mathcal P'}\simeq X_uZ_v\ket{\mathcal P}$. Hence $X_A\ket{\mathcal P'}\simeq X_AX_uZ_v\ket{\mathcal P}\simeq X_AX_uZ_vX_A^{-1}X_A\ket{\mathcal P}$. So, using Identity \eqref{Pauli-conj}, we deduce that  $X_A\ket{\mathcal P'}\simeq X_{Au}Z_{A^{-t}v}X_A\ket{\mathcal P}$ and consequently, $X_A\ket{\mathcal P'}$ is in the LU orbit of
   $X_A\ket{\mathcal P}$, which implies that $|\Delta_4|$ is maximal for the state $X_A\ket{\mathcal P'}$.
   \end{proof}

   \begin{example}
     Let us apply the following sequence of  operations on $\mathcal P_{\mathrm{max}}$ : $\alpha_0\leftarrow\alpha_0+\pi$, $\beta_2\leftarrow-\beta_2$, $\alpha_3\leftarrow -\alpha_3$,
     $\beta_3\leftarrow\beta_3+\pi$. The resulting matrix of parameters is $\mathcal P=\begin{bmatrix}3\pi/2&\pi/2&0\\
  \pi/2&\pi/2&0\\
  \pi/4&-\cos^{-1}\frac{\sqrt 3}{3}&0\\
  -\pi/4&\cos^{-1}\frac{\sqrt 3}{3}+\pi&0\end{bmatrix}$ and $\ket{\mathcal P}=(-\ii Z_0)Z_2(-\ii Y_3)\ket{\mathcal P_{\mathrm{max}}}\simeq X_3Z_0Z_2Z_3\ket{\mathcal P_{\mathrm{max}}}$. Let $u=[0,0,0,1]^t$ and $v=[1,0,1,1]^t$. One has : $A^{(0,1)}_2=\begin{bmatrix}0&1&1&0\\1&0&1&1\\1&1&1&1\\0&1&0&1\end{bmatrix}$, $A^{(0,1)}_2u=[0,1,1,1]^t$,
$\left(A^{(0,1)}_2\right)^{-t}=\begin{bmatrix}1&0&1&0\\1&1&1&1\\0&1&1&1\\1&0&0&1\end{bmatrix}$, $\left(A^{(0,1)}_2\right)^{-t}v=[0,1,0,0]^t$. Hence
$M^{(0,1)}_2\ket{\mathcal P}\simeq X_1X_2X_3Z_1\ket{\psi_{\mathrm{max}}}$.
     \end{example}

    \begin{rem}
      We observe that the matrices $\mathcal P_{\mathrm{max}}$ and $\mathcal P_{\mathrm{max}}'$ are not related by the operations on parameters described in Lemma \ref{P-ops-Pauli}, \textit{i.e.} there does not exist any gate $X_uZ_v$ in the four-qubit Pauli group such that  $\ket{\mathcal P_{\mathrm{max}}'}\simeq X_uZ_v\ket{\mathcal P_{\mathrm{max}}}$.  This implies that the state $\ket{\psi_{\mathrm{max}}}$  and the state  $\ket{\psi_{\mathrm{max}}'}$ define distinct orbits by the action of the four-qubit Pauli group.
      Actually, from Identities  \eqref{P-max} and \eqref{Rz-classical}, one has $\ket{\mathcal P_{\mathrm{max}}'}\simeq P_2P_3\ket{\mathcal P_{\mathrm{max}}}$. Using the method described in Section \ref{method}, we compute a matrix of parameters
      $\mathcal P_{\psi\rightarrow\psi'}=\begin{bmatrix}-\pi/2&-\pi/2&-\pi/2\\
        \pi/2&\pi&\pi\\
        0&\pi/2&\pi\\
        -\pi/2&-\pi/2&-\pi/2\end{bmatrix}$
      and a phase $\phi=-\frac{\pi}{3}$ such that
      $\ket{\psi_{\mathrm{max}}'}=\ee^{\ii\phi}U(\mathcal P_{\psi \rightarrow \psi'})\ket{\psi_{\mathrm{max}}}$.
      Then, using Identities \eqref{Rz-classical} and \eqref{Ry-classical}, we obtain :
      \begin{equation}
        \ket{\psi_{\mathrm{max}}'} \simeq PHP^{\dag}\otimes PX\otimes H\otimes PHP^{\dag}  \ket{\psi_{\mathrm{max}}}\label{LU-psi-to-psi-prime}
      \end{equation}
      This last identity can be checked using the function 
\emph{\texttt{check\_psi\_to\_psi\_prime}} of the Python module.
     

          

    \end{rem}

    \begin{rem}\label{PMAX-def}
Since $M_{3}^{(0,1)}=S_{(023)}M_{2}^{(0,1)}$ , it is easy to see that the set of all matrices of parameters $\mathcal P$ having their last column null such that $M^{(0,1)}_3\ket{\mathcal P}$ maximizes $|\Delta_4|$ is equal to the set of all matrices of parameters $\mathcal P$ having their last column null such that $M^{(0,1)}_2\ket{\mathcal P}$ maximizes $|\Delta_4|$. We denote this set by $\param^{(0,1)}$.
\end{rem}

\begin{conj}\label{ops-conj}
      Any matrix in $\param^{(0,1)}$ can be obtain from $\mathcal P_{\mathrm{max}}$ or from $\mathcal P_{\mathrm{max}}'$ by a sequence of the operations on parameters described by Lemma \ref{P-ops-Pauli}. 
    \end{conj}

      \section{All $\cnot$ circuits to reach the maximum of  $|\Delta_4|$\label{cosets}}
      We  generalize the results of the previous section by describing all the four-qubit $\cnot$ circuits that enable to produce a state maximizing $|\Delta_4|$ when they act on the LU orbit of $\ket{0000}$.
      
      \begin{prop}\label{coset-def}
        Let $i,j,k,\ell$ be distinct integers in $\{0,1,2,3\}$, then $M_{k}^{(i,j)}$ and $M_{\ell}^{(i,j)}$ define the same right coset of the subgroup $\swapg[4]$ in $\cnotg[4]$.
        This right coset is denoted by $\overline{(i,j)}$ :
        \begin{equation}
\overline{(i,j)}=\{S_{\sigma}X_{[ij]}X_{[jk]}X_{[ki]}X_{[i\ell]}X_{[\ell j]}\mid \sigma\in\symg[4]\}.
\end{equation}
        \end{prop}

        \begin{proof}
          Let $\sigma$ be the permutation $\begin{pmatrix}0&1&2&3\\i&j&k&\ell\end{pmatrix}$. Using Identity \eqref{Xij-conj-sigma}, we conjugate each member of the equality $M_{3}^{(0,1)}=S_{(023)}M_{2}^{(0,1)}$ (Identity \eqref{same-coset}) by $S_{\sigma}$  and obtain $M_{\ell}^{(i,j)}=S_{(ik\ell)}M_{k}^{(i,j)}$.
        \end{proof}
      
      \begin{prop}
       Let $i,j,i',j'$ in $\{0,1,2,3\}$ such that $i\neq j$ and $i'\neq j'$.
        If $(i,j)$ and $(i',j')$ are distinct couples, then $\overline{(i,j)}$ and $\overline{(i',j')}$ are distinct cosets.
      \end{prop}

      \begin{proof}
        We check that for any $i,j,k,i',j',k'$ in $\{0,1,2,3\}$ ($i,j,k$ distinct and $i',j',k'$ distinct), if $A_{k}^{(i,j)}=\sigma A_{k'}^{(i',j')}$ for some permutation matrix $\sigma$, then $(i,j) = (i',j')$ (function \texttt{check\_distinct\_cosets} of the Python module).
      \end{proof}

\begin{prop}
For any permutation $\sigma$, one has : $S_{\sigma}\overline{(i,j)}S_{\sigma}^{-1}=\overline{(\sigma(i),\sigma(j))}$.
  \end{prop}

  \begin{proof}
Using Identity \eqref{Xij-conj-sigma} one has $S_{\sigma}M_{k}^{(i,j)}S_{\sigma}^{-1}=M_{\sigma(k)}^{(\sigma(i),\sigma(j))}$. The result follows from Proposition \ref{coset-def}.
\end{proof}

Remark \ref{PMAX-def} can be generalized to any coset $\overline{(i,j)}$ and one can define $\param^{(i,j)}$ as being the set of all matrices $\mathcal P$ having their last column such that
  the state $M_{k}^{(i,j)}\ket{\mathcal P}$ maximizes $|\Delta_4|$.
  \begin{prop} Let $i,j$ be distinct integers in $\{0,1,2,3\}$ and $\sigma$ be a permutation such that $\sigma(0)=i$ and $\sigma(1)=j$, then :
    \begin{equation}
\param^{(i,j)}=\sigma \param^{(0,1)}.
      \end{equation}
  \end{prop}

  \begin{proof}
    We prove that $\param^{(i,j)}\subset\sigma \param^{(0,1)}$, the other inclusion being similar. Let $\mathcal P$ be a matrix of parameters in $\param^{(i,j)}$. Then $M_{k}^{(i,j)}\ket{\mathcal P}=\ket{\psi}$ and $\ket{\psi}$ maximizes $|\Delta_4|$, so 
    $S_{\sigma}^{-1}\ket{\psi}=S_{\sigma}^{-1}M_{k}^{(i,j)}\ket{\mathcal P}=S_{\sigma}^{-1}M_{k}^{(i,j)}S_{\sigma}S_{\sigma}^{-1}\ket{\mathcal P}=M_{k'}^{(0,1)}S_{\sigma}^{-1}\ket{\mathcal P}$, where $k'\in\{2,3\}$. We observe that $S_{\sigma}^{-1}\ket{\mathcal P}$ can be rewritten as $\ket{\sigma^{-1}\mathcal P}$ by using Identity \eqref{conj-sigma-U},
    hence $S_{\sigma}^{-1}\ket{\psi}=M_{k'}^{(0,1)}\ket{\sigma^{-1}\mathcal P}$.     
    As $S_{\sigma}^{-1}\ket{\psi}$ maximizes $|\Delta_4|$, there exists $\mathcal P'$ in $\param^{(0,1)}$ such that $\sigma^{-1}\mathcal P=\mathcal P'$, so $\mathcal P=\sigma \mathcal P'$ and we deduce that $\param^{(i,j)}\subset\sigma \param^{(0,1)}$.
    \end{proof}

    Our numerical results based on the use of the random walk heuristic suggest the following conjecture.
    \begin{conj}
      Any four-qubit $\cnot$ circuit capable of maximizing $|\Delta_4|$ by acting on the LU orbit of $\ket{0000}$ belongs to one of the 12 cosets $\overline{(i,j)}$.
      \end{conj}
      \begin{example}Let $\sigma=(013)$.
        The unitary $S_{\sigma}M_{2}^{(0,1)}$ acting on the state $\ket{\mathcal P_{\mathrm{max}}}$ generates the state $S_{(013)}\ket{\psi_{\mathrm{max}}}$ that maximizes $|\Delta_4|$.
        This state  can be produced by the unitary $M_{\sigma(2)}^{(\sigma(0),\sigma(1))}=M_{2}^{(1,3)}$ acting on the state $\ket{\sigma\mathcal P_{\mathrm{max}}}$, where
        $\sigma\mathcal P_{\mathrm{max}}=\sigma\begin{bmatrix}\pi/2&\pi/2&0\\
          \pi/2&\pi/2&0\\
          \pi/4&\cos^{-1}\frac{\sqrt 3}{3}&0\\
          \pi/4&\cos^{-1}\frac{\sqrt 3}{3}&0\end{bmatrix}=
        \begin{bmatrix}  \pi/4&\cos^{-1}\frac{\sqrt 3}{3}&0\\
          \pi/2&\pi/2&0\\
          \pi/4&\cos^{-1}\frac{\sqrt 3}{3}&0\\
          \pi/2&\pi/2&0              
        \end{bmatrix}$.
      \end{example}

\section{Circuits generating the states $\ket{L}$, $\ket{\Phi_5}$ and $\ket{M_{2222}}$\label{generation}}

Let $\ket{\psi}$ be a state in the set $\{\ket{L}, \ket{\Phi_5}, \ket{M_{2222}}\}$. Following the Gour-Wallach conjecture \cite{2012GW} proved by Chen and Djokovic \cite{2013CD}, there exists a matrix $\mathcal P$ of parameters and a phase $\phi$ such that
\begin{equation}
  \ket{\psi}=\ee^{\ii\phi}U(\mathcal{P})\ket{\psi_{\mathrm{max}}}.\label{LU-equation}
\end{equation}
In practice, one has to solve a 16 equations non linear system but its resolution seems to be out of reach of current equation solvers (we used Maple and Python SymPy solvers). However, one can turn the problem of finding a solution of \eqref{LU-equation} into an optimization problem thanks to this simple remark : $\ket{\psi}=\ee^{\ii\phi}U(\mathcal{P})\ket{\psi_{\mathrm{max}}}$ if and only if the sum of the absolute values of the 16 coordinates of $\ket{\psi}-\ee^{\ii\phi}U(\mathcal{P})\ket{L}$ vanishes. Again, we use a random walk on a search space of 13 parameters (the 12 parameters of $\mathcal P$ plus the phase $\phi$) to minimize this sum and obtain an approximate solution of the system (function \texttt{search\_LU\_from\_state1\_to\_state2} of the Python module). From this approximate solution, it is possible to deduce an exact solution. The results are summarized in Figure \ref{psi-to-L}. Finally, by combining these results with those of Proposition \ref{generate-psi}, it is possible to propose simple quantum circuits generating the states $\ket{L}$, $\ket{\Phi_5}$ and $\ket{M_{2222}}$
(Figure \ref{circuits}, function \texttt{check\_circuits} of the Python module).

\begin{figure}

 $\ket{L}=\ee^{-\ii\frac{11\pi}{12}}U(\mathcal P_{\psi \rightarrow L})\ket{\psi_{\mathrm{max}}}$, where $\mathcal P_{\psi \rightarrow L}=
\begin{bmatrix}
  \pi&\pi/2&-\pi/2\\
  \pi/2&-\pi/2&\pi\\
  0&\pi&\pi\\
  0&-\pi/2&\pi/2
\end{bmatrix}$\medskip

 $\ket{\Phi_5}=\ee^{-\ii\frac{7\pi}{12}}U(\mathcal P_{\psi \rightarrow \Phi_5})\ket{\psi_{\mathrm{max}}}$, where $\mathcal P_{\psi \rightarrow \Phi_5}=
\begin{bmatrix}
  -\pi/3&\theta-\pi&-3\pi/4\\
  \pi/3&\theta&3\pi/4\\
  \pi&\theta&3\pi/4\\
  2\pi/3&\pi-\theta&\pi/4
\end{bmatrix}$

and $\theta=\cos^{-1}\frac{\sqrt 3}{3}$ \medskip


$\ket{M_{2222}}=\ee^{\ii\frac{5\pi}{12}}U(\mathcal P_{\psi \rightarrow M_{2222}})\ket{\psi_{\mathrm{max}}}$, where $\mathcal P_{\psi \rightarrow M_{2222}}=
\begin{bmatrix}
  \pi/2&\pi/4&0\\
  -\pi/2&-\pi/4&\pi/2\\
  0&-\pi/2&\pi/4\\
  \pi/2&3\pi/4&\pi
\end{bmatrix}$

\caption{LU operators generating the states $\ket{L}$, $\ket{\Phi_5}$ and $\ket{M_{2222}}$, from the state $\ket{\psi_{\mathrm{max}}}$.\label{psi-to-L}}
\end{figure}

\begin{figure}
  \vspace*{3mm}
  
  \begin{center}
  \fbox{$\theta=\cos^{-1}\frac{\sqrt 3}{3}$}
\end{center}
  
\raisebox{15mm}{\hspace{2.5mm}$\ee^{-\ii \frac{7\pi}{12}}\ket{L}=$}
\begin{tikzpicture}[scale=1.1600000,x=1pt,y=1pt]
\filldraw[color=white] (0.000000, -9.500000) rectangle (224.000000, 66.500000);
\draw[color=black] (0.000000,57.000000) -- (224.000000,57.000000);
\draw[color=black] (0.000000,57.000000) node[left] {$\ket{0}$};
\draw[color=black] (0.000000,38.000000) -- (224.000000,38.000000);
\draw[color=black] (0.000000,38.000000) node[left] {$\ket{0}$};
\draw[color=black] (0.000000,19.000000) -- (224.000000,19.000000);
\draw[color=black] (0.000000,19.000000) node[left] {$\ket{0}$};
\draw[color=black] (0.000000,0.000000) -- (224.000000,0.000000);
\draw[color=black] (0.000000,0.000000) node[left] {$\ket{0}$};
\begin{scope}
\draw[fill=white] (19.000000, 57.000000) +(-45.000000:8.485281pt and 8.485281pt) -- +(45.000000:8.485281pt and 8.485281pt) -- +(135.000000:8.485281pt and 8.485281pt) -- +(225.000000:8.485281pt and 8.485281pt) -- cycle;
\clip (19.000000, 57.000000) +(-45.000000:8.485281pt and 8.485281pt) -- +(45.000000:8.485281pt and 8.485281pt) -- +(135.000000:8.485281pt and 8.485281pt) -- +(225.000000:8.485281pt and 8.485281pt) -- cycle;
\draw (19.000000, 57.000000) node {$H$};
\end{scope}
\begin{scope}
\draw[fill=white] (19.000000, 38.000000) +(-45.000000:8.485281pt and 8.485281pt) -- +(45.000000:8.485281pt and 8.485281pt) -- +(135.000000:8.485281pt and 8.485281pt) -- +(225.000000:8.485281pt and 8.485281pt) -- cycle;
\clip (19.000000, 38.000000) +(-45.000000:8.485281pt and 8.485281pt) -- +(45.000000:8.485281pt and 8.485281pt) -- +(135.000000:8.485281pt and 8.485281pt) -- +(225.000000:8.485281pt and 8.485281pt) -- cycle;
\draw (19.000000, 38.000000) node {$H$};
\end{scope}
\begin{scope}
\draw[fill=white] (19.000000, 19.000000) +(-45.000000:18.384776pt and 10.606602pt) -- +(45.000000:18.384776pt and 10.606602pt) -- +(135.000000:18.384776pt and 10.606602pt) -- +(225.000000:18.384776pt and 10.606602pt) -- cycle;
\clip (19.000000, 19.000000) +(-45.000000:18.384776pt and 10.606602pt) -- +(45.000000:18.384776pt and 10.606602pt) -- +(135.000000:18.384776pt and 10.606602pt) -- +(225.000000:18.384776pt and 10.606602pt) -- cycle;
\draw (19.000000, 19.000000) node {$R_y(\theta)$};
\end{scope}
\begin{scope}
\draw[fill=white] (19.000000, -0.000000) +(-45.000000:18.384776pt and 10.606602pt) -- +(45.000000:18.384776pt and 10.606602pt) -- +(135.000000:18.384776pt and 10.606602pt) -- +(225.000000:18.384776pt and 10.606602pt) -- cycle;
\clip (19.000000, -0.000000) +(-45.000000:18.384776pt and 10.606602pt) -- +(45.000000:18.384776pt and 10.606602pt) -- +(135.000000:18.384776pt and 10.606602pt) -- +(225.000000:18.384776pt and 10.606602pt) -- cycle;
\draw (19.000000, -0.000000) node {$R_y(\theta)$};
\end{scope}
\begin{scope}
\draw[fill=white] (50.000000, 57.000000) +(-45.000000:8.485281pt and 8.485281pt) -- +(45.000000:8.485281pt and 8.485281pt) -- +(135.000000:8.485281pt and 8.485281pt) -- +(225.000000:8.485281pt and 8.485281pt) -- cycle;
\clip (50.000000, 57.000000) +(-45.000000:8.485281pt and 8.485281pt) -- +(45.000000:8.485281pt and 8.485281pt) -- +(135.000000:8.485281pt and 8.485281pt) -- +(225.000000:8.485281pt and 8.485281pt) -- cycle;
\draw (50.000000, 57.000000) node {$P$};
\end{scope}
\begin{scope}
\draw[fill=white] (50.000000, 38.000000) +(-45.000000:8.485281pt and 8.485281pt) -- +(45.000000:8.485281pt and 8.485281pt) -- +(135.000000:8.485281pt and 8.485281pt) -- +(225.000000:8.485281pt and 8.485281pt) -- cycle;
\clip (50.000000, 38.000000) +(-45.000000:8.485281pt and 8.485281pt) -- +(45.000000:8.485281pt and 8.485281pt) -- +(135.000000:8.485281pt and 8.485281pt) -- +(225.000000:8.485281pt and 8.485281pt) -- cycle;
\draw (50.000000, 38.000000) node {$P$};
\end{scope}
\begin{scope}
\draw[fill=white] (50.000000, 19.000000) +(-45.000000:8.485281pt and 8.485281pt) -- +(45.000000:8.485281pt and 8.485281pt) -- +(135.000000:8.485281pt and 8.485281pt) -- +(225.000000:8.485281pt and 8.485281pt) -- cycle;
\clip (50.000000, 19.000000) +(-45.000000:8.485281pt and 8.485281pt) -- +(45.000000:8.485281pt and 8.485281pt) -- +(135.000000:8.485281pt and 8.485281pt) -- +(225.000000:8.485281pt and 8.485281pt) -- cycle;
\draw (50.000000, 19.000000) node {$T$};
\end{scope}
\begin{scope}
\draw[fill=white] (50.000000, -0.000000) +(-45.000000:8.485281pt and 8.485281pt) -- +(45.000000:8.485281pt and 8.485281pt) -- +(135.000000:8.485281pt and 8.485281pt) -- +(225.000000:8.485281pt and 8.485281pt) -- cycle;
\clip (50.000000, -0.000000) +(-45.000000:8.485281pt and 8.485281pt) -- +(45.000000:8.485281pt and 8.485281pt) -- +(135.000000:8.485281pt and 8.485281pt) -- +(225.000000:8.485281pt and 8.485281pt) -- cycle;
\draw (50.000000, -0.000000) node {$T$};
\end{scope}
\draw (71.000000,38.000000) -- (71.000000,0.000000);
\begin{scope}
\draw[fill=white] (71.000000, 0.000000) circle(3.000000pt);
\clip (71.000000, 0.000000) circle(3.000000pt);
\draw (68.000000, 0.000000) -- (74.000000, 0.000000);
\draw (71.000000, -3.000000) -- (71.000000, 3.000000);
\end{scope}
\filldraw (71.000000, 38.000000) circle(1.500000pt);
\draw (89.000000,57.000000) -- (89.000000,0.000000);
\begin{scope}
\draw[fill=white] (89.000000, 57.000000) circle(3.000000pt);
\clip (89.000000, 57.000000) circle(3.000000pt);
\draw (86.000000, 57.000000) -- (92.000000, 57.000000);
\draw (89.000000, 54.000000) -- (89.000000, 60.000000);
\end{scope}
\filldraw (89.000000, 0.000000) circle(1.500000pt);
\draw (107.000000,57.000000) -- (107.000000,19.000000);
\begin{scope}
\draw[fill=white] (107.000000, 19.000000) circle(3.000000pt);
\clip (107.000000, 19.000000) circle(3.000000pt);
\draw (104.000000, 19.000000) -- (110.000000, 19.000000);
\draw (107.000000, 16.000000) -- (107.000000, 22.000000);
\end{scope}
\filldraw (107.000000, 57.000000) circle(1.500000pt);
\draw (125.000000,38.000000) -- (125.000000,19.000000);
\begin{scope}
\draw[fill=white] (125.000000, 38.000000) circle(3.000000pt);
\clip (125.000000, 38.000000) circle(3.000000pt);
\draw (122.000000, 38.000000) -- (128.000000, 38.000000);
\draw (125.000000, 35.000000) -- (125.000000, 41.000000);
\end{scope}
\filldraw (125.000000, 19.000000) circle(1.500000pt);
\draw (143.000000,57.000000) -- (143.000000,38.000000);
\begin{scope}
\draw[fill=white] (143.000000, 57.000000) circle(3.000000pt);
\clip (143.000000, 57.000000) circle(3.000000pt);
\draw (140.000000, 57.000000) -- (146.000000, 57.000000);
\draw (143.000000, 54.000000) -- (143.000000, 60.000000);
\end{scope}
\filldraw (143.000000, 38.000000) circle(1.500000pt);
\begin{scope}
\draw[fill=white] (164.000000, 57.000000) +(-45.000000:8.485281pt and 8.485281pt) -- +(45.000000:8.485281pt and 8.485281pt) -- +(135.000000:8.485281pt and 8.485281pt) -- +(225.000000:8.485281pt and 8.485281pt) -- cycle;
\clip (164.000000, 57.000000) +(-45.000000:8.485281pt and 8.485281pt) -- +(45.000000:8.485281pt and 8.485281pt) -- +(135.000000:8.485281pt and 8.485281pt) -- +(225.000000:8.485281pt and 8.485281pt) -- cycle;
\draw (164.000000, 57.000000) node {$P$};
\end{scope}
\begin{scope}
\draw[fill=white] (164.000000, 38.000000) +(-45.000000:8.485281pt and 8.485281pt) -- +(45.000000:8.485281pt and 8.485281pt) -- +(135.000000:8.485281pt and 8.485281pt) -- +(225.000000:8.485281pt and 8.485281pt) -- cycle;
\clip (164.000000, 38.000000) +(-45.000000:8.485281pt and 8.485281pt) -- +(45.000000:8.485281pt and 8.485281pt) -- +(135.000000:8.485281pt and 8.485281pt) -- +(225.000000:8.485281pt and 8.485281pt) -- cycle;
\draw (164.000000, 38.000000) node {$Z$};
\end{scope}
\begin{scope}
\draw[fill=white] (164.000000, 19.000000) +(-45.000000:8.485281pt and 8.485281pt) -- +(45.000000:8.485281pt and 8.485281pt) -- +(135.000000:8.485281pt and 8.485281pt) -- +(225.000000:8.485281pt and 8.485281pt) -- cycle;
\clip (164.000000, 19.000000) +(-45.000000:8.485281pt and 8.485281pt) -- +(45.000000:8.485281pt and 8.485281pt) -- +(135.000000:8.485281pt and 8.485281pt) -- +(225.000000:8.485281pt and 8.485281pt) -- cycle;
\draw (164.000000, 19.000000) node {$X$};
\end{scope}
\begin{scope}
\draw[fill=white] (164.000000, -0.000000) +(-45.000000:8.485281pt and 8.485281pt) -- +(45.000000:8.485281pt and 8.485281pt) -- +(135.000000:8.485281pt and 8.485281pt) -- +(225.000000:8.485281pt and 8.485281pt) -- cycle;
\clip (164.000000, -0.000000) +(-45.000000:8.485281pt and 8.485281pt) -- +(45.000000:8.485281pt and 8.485281pt) -- +(135.000000:8.485281pt and 8.485281pt) -- +(225.000000:8.485281pt and 8.485281pt) -- cycle;
\draw (164.000000, -0.000000) node {$P$};
\end{scope}
\begin{scope}
\draw[fill=white] (188.000000, 57.000000) +(-45.000000:8.485281pt and 8.485281pt) -- +(45.000000:8.485281pt and 8.485281pt) -- +(135.000000:8.485281pt and 8.485281pt) -- +(225.000000:8.485281pt and 8.485281pt) -- cycle;
\clip (188.000000, 57.000000) +(-45.000000:8.485281pt and 8.485281pt) -- +(45.000000:8.485281pt and 8.485281pt) -- +(135.000000:8.485281pt and 8.485281pt) -- +(225.000000:8.485281pt and 8.485281pt) -- cycle;
\draw (188.000000, 57.000000) node {$H$};
\end{scope}
\begin{scope}
\draw[fill=white] (188.000000, 38.000000) +(-45.000000:8.485281pt and 8.485281pt) -- +(45.000000:8.485281pt and 8.485281pt) -- +(135.000000:8.485281pt and 8.485281pt) -- +(225.000000:8.485281pt and 8.485281pt) -- cycle;
\clip (188.000000, 38.000000) +(-45.000000:8.485281pt and 8.485281pt) -- +(45.000000:8.485281pt and 8.485281pt) -- +(135.000000:8.485281pt and 8.485281pt) -- +(225.000000:8.485281pt and 8.485281pt) -- cycle;
\draw (188.000000, 38.000000) node {$H$};
\end{scope}
\begin{scope}
\draw[fill=white] (188.000000, -0.000000) +(-45.000000:8.485281pt and 8.485281pt) -- +(45.000000:8.485281pt and 8.485281pt) -- +(135.000000:8.485281pt and 8.485281pt) -- +(225.000000:8.485281pt and 8.485281pt) -- cycle;
\clip (188.000000, -0.000000) +(-45.000000:8.485281pt and 8.485281pt) -- +(45.000000:8.485281pt and 8.485281pt) -- +(135.000000:8.485281pt and 8.485281pt) -- +(225.000000:8.485281pt and 8.485281pt) -- cycle;
\draw (188.000000, -0.000000) node {$H$};
\end{scope}
\begin{scope}
\draw[fill=white] (212.000000, 57.000000) +(-45.000000:8.485281pt and 8.485281pt) -- +(45.000000:8.485281pt and 8.485281pt) -- +(135.000000:8.485281pt and 8.485281pt) -- +(225.000000:8.485281pt and 8.485281pt) -- cycle;
\clip (212.000000, 57.000000) +(-45.000000:8.485281pt and 8.485281pt) -- +(45.000000:8.485281pt and 8.485281pt) -- +(135.000000:8.485281pt and 8.485281pt) -- +(225.000000:8.485281pt and 8.485281pt) -- cycle;
\draw (212.000000, 57.000000) node {$Z$};
\end{scope}
\begin{scope}
\draw[fill=white] (212.000000, 38.000000) +(-45.000000:8.485281pt and 8.485281pt) -- +(45.000000:8.485281pt and 8.485281pt) -- +(135.000000:8.485281pt and 8.485281pt) -- +(225.000000:8.485281pt and 8.485281pt) -- cycle;
\clip (212.000000, 38.000000) +(-45.000000:8.485281pt and 8.485281pt) -- +(45.000000:8.485281pt and 8.485281pt) -- +(135.000000:8.485281pt and 8.485281pt) -- +(225.000000:8.485281pt and 8.485281pt) -- cycle;
\draw (212.000000, 38.000000) node {$P^{\dag}$};
\end{scope}
\begin{scope}
\draw[fill=white] (212.000000, -0.000000) +(-45.000000:8.485281pt and 8.485281pt) -- +(45.000000:8.485281pt and 8.485281pt) -- +(135.000000:8.485281pt and 8.485281pt) -- +(225.000000:8.485281pt and 8.485281pt) -- cycle;
\clip (212.000000, -0.000000) +(-45.000000:8.485281pt and 8.485281pt) -- +(45.000000:8.485281pt and 8.485281pt) -- +(135.000000:8.485281pt and 8.485281pt) -- +(225.000000:8.485281pt and 8.485281pt) -- cycle;
\draw (212.000000, -0.000000) node {$Z$};
\end{scope}
\end{tikzpicture}

\begin{center}
\raisebox{15mm}{$\ee^{-\ii \frac{\pi}{6}}\ket{\Phi_5}=$}
\begin{tikzpicture}[scale=1.1600000,x=1pt,y=1pt]
\filldraw[color=white] (0.000000, -9.500000) rectangle (291.000000, 66.500000);
\draw[color=black] (0.000000,57.000000) -- (291.000000,57.000000);
\draw[color=black] (0.000000,57.000000) node[left] {$\ket{0}$};
\draw[color=black] (0.000000,38.000000) -- (291.000000,38.000000);
\draw[color=black] (0.000000,38.000000) node[left] {$\ket{0}$};
\draw[color=black] (0.000000,19.000000) -- (291.000000,19.000000);
\draw[color=black] (0.000000,19.000000) node[left] {$\ket{0}$};
\draw[color=black] (0.000000,0.000000) -- (291.000000,0.000000);
\draw[color=black] (0.000000,0.000000) node[left] {$\ket{0}$};
\begin{scope}
\draw[fill=white] (19.000000, 57.000000) +(-45.000000:8.485281pt and 8.485281pt) -- +(45.000000:8.485281pt and 8.485281pt) -- +(135.000000:8.485281pt and 8.485281pt) -- +(225.000000:8.485281pt and 8.485281pt) -- cycle;
\clip (19.000000, 57.000000) +(-45.000000:8.485281pt and 8.485281pt) -- +(45.000000:8.485281pt and 8.485281pt) -- +(135.000000:8.485281pt and 8.485281pt) -- +(225.000000:8.485281pt and 8.485281pt) -- cycle;
\draw (19.000000, 57.000000) node {$H$};
\end{scope}
\begin{scope}
\draw[fill=white] (19.000000, 38.000000) +(-45.000000:8.485281pt and 8.485281pt) -- +(45.000000:8.485281pt and 8.485281pt) -- +(135.000000:8.485281pt and 8.485281pt) -- +(225.000000:8.485281pt and 8.485281pt) -- cycle;
\clip (19.000000, 38.000000) +(-45.000000:8.485281pt and 8.485281pt) -- +(45.000000:8.485281pt and 8.485281pt) -- +(135.000000:8.485281pt and 8.485281pt) -- +(225.000000:8.485281pt and 8.485281pt) -- cycle;
\draw (19.000000, 38.000000) node {$H$};
\end{scope}
\begin{scope}
\draw[fill=white] (19.000000, 19.000000) +(-45.000000:18.384776pt and 10.606602pt) -- +(45.000000:18.384776pt and 10.606602pt) -- +(135.000000:18.384776pt and 10.606602pt) -- +(225.000000:18.384776pt and 10.606602pt) -- cycle;
\clip (19.000000, 19.000000) +(-45.000000:18.384776pt and 10.606602pt) -- +(45.000000:18.384776pt and 10.606602pt) -- +(135.000000:18.384776pt and 10.606602pt) -- +(225.000000:18.384776pt and 10.606602pt) -- cycle;
\draw (19.000000, 19.000000) node {$R_y(\theta)$};
\end{scope}
\begin{scope}
\draw[fill=white] (19.000000, -0.000000) +(-45.000000:18.384776pt and 10.606602pt) -- +(45.000000:18.384776pt and 10.606602pt) -- +(135.000000:18.384776pt and 10.606602pt) -- +(225.000000:18.384776pt and 10.606602pt) -- cycle;
\clip (19.000000, -0.000000) +(-45.000000:18.384776pt and 10.606602pt) -- +(45.000000:18.384776pt and 10.606602pt) -- +(135.000000:18.384776pt and 10.606602pt) -- +(225.000000:18.384776pt and 10.606602pt) -- cycle;
\draw (19.000000, -0.000000) node {$R_y(\theta)$};
\end{scope}
\begin{scope}
\draw[fill=white] (50.000000, 57.000000) +(-45.000000:8.485281pt and 8.485281pt) -- +(45.000000:8.485281pt and 8.485281pt) -- +(135.000000:8.485281pt and 8.485281pt) -- +(225.000000:8.485281pt and 8.485281pt) -- cycle;
\clip (50.000000, 57.000000) +(-45.000000:8.485281pt and 8.485281pt) -- +(45.000000:8.485281pt and 8.485281pt) -- +(135.000000:8.485281pt and 8.485281pt) -- +(225.000000:8.485281pt and 8.485281pt) -- cycle;
\draw (50.000000, 57.000000) node {$P$};
\end{scope}
\begin{scope}
\draw[fill=white] (50.000000, 38.000000) +(-45.000000:8.485281pt and 8.485281pt) -- +(45.000000:8.485281pt and 8.485281pt) -- +(135.000000:8.485281pt and 8.485281pt) -- +(225.000000:8.485281pt and 8.485281pt) -- cycle;
\clip (50.000000, 38.000000) +(-45.000000:8.485281pt and 8.485281pt) -- +(45.000000:8.485281pt and 8.485281pt) -- +(135.000000:8.485281pt and 8.485281pt) -- +(225.000000:8.485281pt and 8.485281pt) -- cycle;
\draw (50.000000, 38.000000) node {$P$};
\end{scope}
\begin{scope}
\draw[fill=white] (50.000000, 19.000000) +(-45.000000:8.485281pt and 8.485281pt) -- +(45.000000:8.485281pt and 8.485281pt) -- +(135.000000:8.485281pt and 8.485281pt) -- +(225.000000:8.485281pt and 8.485281pt) -- cycle;
\clip (50.000000, 19.000000) +(-45.000000:8.485281pt and 8.485281pt) -- +(45.000000:8.485281pt and 8.485281pt) -- +(135.000000:8.485281pt and 8.485281pt) -- +(225.000000:8.485281pt and 8.485281pt) -- cycle;
\draw (50.000000, 19.000000) node {$T$};
\end{scope}
\begin{scope}
\draw[fill=white] (50.000000, -0.000000) +(-45.000000:8.485281pt and 8.485281pt) -- +(45.000000:8.485281pt and 8.485281pt) -- +(135.000000:8.485281pt and 8.485281pt) -- +(225.000000:8.485281pt and 8.485281pt) -- cycle;
\clip (50.000000, -0.000000) +(-45.000000:8.485281pt and 8.485281pt) -- +(45.000000:8.485281pt and 8.485281pt) -- +(135.000000:8.485281pt and 8.485281pt) -- +(225.000000:8.485281pt and 8.485281pt) -- cycle;
\draw (50.000000, -0.000000) node {$T$};
\end{scope}
\draw (71.000000,38.000000) -- (71.000000,0.000000);
\begin{scope}
\draw[fill=white] (71.000000, 0.000000) circle(3.000000pt);
\clip (71.000000, 0.000000) circle(3.000000pt);
\draw (68.000000, 0.000000) -- (74.000000, 0.000000);
\draw (71.000000, -3.000000) -- (71.000000, 3.000000);
\end{scope}
\filldraw (71.000000, 38.000000) circle(1.500000pt);
\draw (89.000000,57.000000) -- (89.000000,0.000000);
\begin{scope}
\draw[fill=white] (89.000000, 57.000000) circle(3.000000pt);
\clip (89.000000, 57.000000) circle(3.000000pt);
\draw (86.000000, 57.000000) -- (92.000000, 57.000000);
\draw (89.000000, 54.000000) -- (89.000000, 60.000000);
\end{scope}
\filldraw (89.000000, 0.000000) circle(1.500000pt);
\draw (107.000000,57.000000) -- (107.000000,19.000000);
\begin{scope}
\draw[fill=white] (107.000000, 19.000000) circle(3.000000pt);
\clip (107.000000, 19.000000) circle(3.000000pt);
\draw (104.000000, 19.000000) -- (110.000000, 19.000000);
\draw (107.000000, 16.000000) -- (107.000000, 22.000000);
\end{scope}
\filldraw (107.000000, 57.000000) circle(1.500000pt);
\draw (125.000000,38.000000) -- (125.000000,19.000000);
\begin{scope}
\draw[fill=white] (125.000000, 38.000000) circle(3.000000pt);
\clip (125.000000, 38.000000) circle(3.000000pt);
\draw (122.000000, 38.000000) -- (128.000000, 38.000000);
\draw (125.000000, 35.000000) -- (125.000000, 41.000000);
\end{scope}
\filldraw (125.000000, 19.000000) circle(1.500000pt);
\draw (143.000000,57.000000) -- (143.000000,38.000000);
\begin{scope}
\draw[fill=white] (143.000000, 57.000000) circle(3.000000pt);
\clip (143.000000, 57.000000) circle(3.000000pt);
\draw (140.000000, 57.000000) -- (146.000000, 57.000000);
\draw (143.000000, 54.000000) -- (143.000000, 60.000000);
\end{scope}
\filldraw (143.000000, 38.000000) circle(1.500000pt);
\begin{scope}
\draw[fill=white] (164.000000, 57.000000) +(-45.000000:8.485281pt and 8.485281pt) -- +(45.000000:8.485281pt and 8.485281pt) -- +(135.000000:8.485281pt and 8.485281pt) -- +(225.000000:8.485281pt and 8.485281pt) -- cycle;
\clip (164.000000, 57.000000) +(-45.000000:8.485281pt and 8.485281pt) -- +(45.000000:8.485281pt and 8.485281pt) -- +(135.000000:8.485281pt and 8.485281pt) -- +(225.000000:8.485281pt and 8.485281pt) -- cycle;
\draw (164.000000, 57.000000) node {$T$};
\end{scope}
\begin{scope}
\draw[fill=white] (164.000000, 38.000000) +(-45.000000:8.485281pt and 8.485281pt) -- +(45.000000:8.485281pt and 8.485281pt) -- +(135.000000:8.485281pt and 8.485281pt) -- +(225.000000:8.485281pt and 8.485281pt) -- cycle;
\clip (164.000000, 38.000000) +(-45.000000:8.485281pt and 8.485281pt) -- +(45.000000:8.485281pt and 8.485281pt) -- +(135.000000:8.485281pt and 8.485281pt) -- +(225.000000:8.485281pt and 8.485281pt) -- cycle;
\draw (164.000000, 38.000000) node {$T$};
\end{scope}
\begin{scope}
\draw[fill=white] (164.000000, 19.000000) +(-45.000000:8.485281pt and 8.485281pt) -- +(45.000000:8.485281pt and 8.485281pt) -- +(135.000000:8.485281pt and 8.485281pt) -- +(225.000000:8.485281pt and 8.485281pt) -- cycle;
\clip (164.000000, 19.000000) +(-45.000000:8.485281pt and 8.485281pt) -- +(45.000000:8.485281pt and 8.485281pt) -- +(135.000000:8.485281pt and 8.485281pt) -- +(225.000000:8.485281pt and 8.485281pt) -- cycle;
\draw (164.000000, 19.000000) node {$T$};
\end{scope}
\begin{scope}
\draw[fill=white] (164.000000, -0.000000) +(-45.000000:8.485281pt and 8.485281pt) -- +(45.000000:8.485281pt and 8.485281pt) -- +(135.000000:8.485281pt and 8.485281pt) -- +(225.000000:8.485281pt and 8.485281pt) -- cycle;
\clip (164.000000, -0.000000) +(-45.000000:8.485281pt and 8.485281pt) -- +(45.000000:8.485281pt and 8.485281pt) -- +(135.000000:8.485281pt and 8.485281pt) -- +(225.000000:8.485281pt and 8.485281pt) -- cycle;
\draw (164.000000, -0.000000) node {$T$};
\end{scope}
\begin{scope}
\draw[fill=white] (188.000000, 57.000000) +(-45.000000:8.485281pt and 8.485281pt) -- +(45.000000:8.485281pt and 8.485281pt) -- +(135.000000:8.485281pt and 8.485281pt) -- +(225.000000:8.485281pt and 8.485281pt) -- cycle;
\clip (188.000000, 57.000000) +(-45.000000:8.485281pt and 8.485281pt) -- +(45.000000:8.485281pt and 8.485281pt) -- +(135.000000:8.485281pt and 8.485281pt) -- +(225.000000:8.485281pt and 8.485281pt) -- cycle;
\draw (188.000000, 57.000000) node {$X$};
\end{scope}
\begin{scope}
\draw[fill=white] (188.000000, 38.000000) +(-45.000000:8.485281pt and 8.485281pt) -- +(45.000000:8.485281pt and 8.485281pt) -- +(135.000000:8.485281pt and 8.485281pt) -- +(225.000000:8.485281pt and 8.485281pt) -- cycle;
\clip (188.000000, 38.000000) +(-45.000000:8.485281pt and 8.485281pt) -- +(45.000000:8.485281pt and 8.485281pt) -- +(135.000000:8.485281pt and 8.485281pt) -- +(225.000000:8.485281pt and 8.485281pt) -- cycle;
\draw (188.000000, 38.000000) node {$P$};
\end{scope}
\begin{scope}
\draw[fill=white] (188.000000, 19.000000) +(-45.000000:8.485281pt and 8.485281pt) -- +(45.000000:8.485281pt and 8.485281pt) -- +(135.000000:8.485281pt and 8.485281pt) -- +(225.000000:8.485281pt and 8.485281pt) -- cycle;
\clip (188.000000, 19.000000) +(-45.000000:8.485281pt and 8.485281pt) -- +(45.000000:8.485281pt and 8.485281pt) -- +(135.000000:8.485281pt and 8.485281pt) -- +(225.000000:8.485281pt and 8.485281pt) -- cycle;
\draw (188.000000, 19.000000) node {$P$};
\end{scope}
\begin{scope}
\draw[fill=white] (188.000000, -0.000000) +(-45.000000:8.485281pt and 8.485281pt) -- +(45.000000:8.485281pt and 8.485281pt) -- +(135.000000:8.485281pt and 8.485281pt) -- +(225.000000:8.485281pt and 8.485281pt) -- cycle;
\clip (188.000000, -0.000000) +(-45.000000:8.485281pt and 8.485281pt) -- +(45.000000:8.485281pt and 8.485281pt) -- +(135.000000:8.485281pt and 8.485281pt) -- +(225.000000:8.485281pt and 8.485281pt) -- cycle;
\draw (188.000000, -0.000000) node {$Y$};
\end{scope}
\begin{scope}
\draw[fill=white] (222.500000, 57.000000) +(-45.000000:18.384776pt and 10.606602pt) -- +(45.000000:18.384776pt and 10.606602pt) -- +(135.000000:18.384776pt and 10.606602pt) -- +(225.000000:18.384776pt and 10.606602pt) -- cycle;
\clip (222.500000, 57.000000) +(-45.000000:18.384776pt and 10.606602pt) -- +(45.000000:18.384776pt and 10.606602pt) -- +(135.000000:18.384776pt and 10.606602pt) -- +(225.000000:18.384776pt and 10.606602pt) -- cycle;
\draw (222.500000, 57.000000) node {$R_y(\theta)$};
\end{scope}
\begin{scope}
\draw[fill=white] (222.500000, 38.000000) +(-45.000000:18.384776pt and 10.606602pt) -- +(45.000000:18.384776pt and 10.606602pt) -- +(135.000000:18.384776pt and 10.606602pt) -- +(225.000000:18.384776pt and 10.606602pt) -- cycle;
\clip (222.500000, 38.000000) +(-45.000000:18.384776pt and 10.606602pt) -- +(45.000000:18.384776pt and 10.606602pt) -- +(135.000000:18.384776pt and 10.606602pt) -- +(225.000000:18.384776pt and 10.606602pt) -- cycle;
\draw (222.500000, 38.000000) node {$R_y(\theta)$};
\end{scope}
\begin{scope}
\draw[fill=white] (222.500000, 19.000000) +(-45.000000:18.384776pt and 10.606602pt) -- +(45.000000:18.384776pt and 10.606602pt) -- +(135.000000:18.384776pt and 10.606602pt) -- +(225.000000:18.384776pt and 10.606602pt) -- cycle;
\clip (222.500000, 19.000000) +(-45.000000:18.384776pt and 10.606602pt) -- +(45.000000:18.384776pt and 10.606602pt) -- +(135.000000:18.384776pt and 10.606602pt) -- +(225.000000:18.384776pt and 10.606602pt) -- cycle;
\draw (222.500000, 19.000000) node {$R_y(\theta)$};
\end{scope}
\begin{scope}
\draw[fill=white] (222.500000, -0.000000) +(-45.000000:23.334524pt and 10.606602pt) -- +(45.000000:23.334524pt and 10.606602pt) -- +(135.000000:23.334524pt and 10.606602pt) -- +(225.000000:23.334524pt and 10.606602pt) -- cycle;
\clip (222.500000, -0.000000) +(-45.000000:23.334524pt and 10.606602pt) -- +(45.000000:23.334524pt and 10.606602pt) -- +(135.000000:23.334524pt and 10.606602pt) -- +(225.000000:23.334524pt and 10.606602pt) -- cycle;
\draw (222.500000, -0.000000) node {$R_y(-\theta)$};
\end{scope}
\begin{scope}
\draw[fill=white] (268.000000, 57.000000) +(-45.000000:24.041631pt and 10.606602pt) -- +(45.000000:24.041631pt and 10.606602pt) -- +(135.000000:24.041631pt and 10.606602pt) -- +(225.000000:24.041631pt and 10.606602pt) -- cycle;
\clip (268.000000, 57.000000) +(-45.000000:24.041631pt and 10.606602pt) -- +(45.000000:24.041631pt and 10.606602pt) -- +(135.000000:24.041631pt and 10.606602pt) -- +(225.000000:24.041631pt and 10.606602pt) -- cycle;
\draw (268.000000, 57.000000) node {$R_z(-\frac{\pi}{3})$};
\end{scope}
\begin{scope}
\draw[fill=white] (268.000000, 38.000000) +(-45.000000:21.920310pt and 11.313708pt) -- +(45.000000:21.920310pt and 11.313708pt) -- +(135.000000:21.920310pt and 11.313708pt) -- +(225.000000:21.920310pt and 11.313708pt) -- cycle;
\clip (268.000000, 38.000000) +(-45.000000:21.920310pt and 11.313708pt) -- +(45.000000:21.920310pt and 11.313708pt) -- +(135.000000:21.920310pt and 11.313708pt) -- +(225.000000:21.920310pt and 11.313708pt) -- cycle;
\draw (268.000000, 38.000000) node {$R_z(\frac{\pi}{3})$};
\end{scope}
\begin{scope}
\draw[fill=white] (268.000000, 19.000000) +(-45.000000:8.485281pt and 8.485281pt) -- +(45.000000:8.485281pt and 8.485281pt) -- +(135.000000:8.485281pt and 8.485281pt) -- +(225.000000:8.485281pt and 8.485281pt) -- cycle;
\clip (268.000000, 19.000000) +(-45.000000:8.485281pt and 8.485281pt) -- +(45.000000:8.485281pt and 8.485281pt) -- +(135.000000:8.485281pt and 8.485281pt) -- +(225.000000:8.485281pt and 8.485281pt) -- cycle;
\draw (268.000000, 19.000000) node {$Z$};
\end{scope}
\begin{scope}
\draw[fill=white] (268.000000, -0.000000) +(-45.000000:21.920310pt and 10.606602pt) -- +(45.000000:21.920310pt and 10.606602pt) -- +(135.000000:21.920310pt and 10.606602pt) -- +(225.000000:21.920310pt and 10.606602pt) -- cycle;
\clip (268.000000, -0.000000) +(-45.000000:21.920310pt and 10.606602pt) -- +(45.000000:21.920310pt and 10.606602pt) -- +(135.000000:21.920310pt and 10.606602pt) -- +(225.000000:21.920310pt and 10.606602pt) -- cycle;
\draw (268.000000, -0.000000) node {$R_z(\frac{2\pi}{3})$};
\end{scope}
\end{tikzpicture}

\end{center}

\raisebox{15mm}{$\ee^{-\ii \frac{13\pi}{24}}\ket{M_{2222}}=$}
\begin{tikzpicture}[scale=1.160000,x=1pt,y=1pt]
\filldraw[color=white] (0.000000, -9.500000) rectangle (247.000000, 66.500000);
\draw[color=black] (0.000000,57.000000) -- (247.000000,57.000000);
\draw[color=black] (0.000000,57.000000) node[left] {$\ket{0}$};
\draw[color=black] (0.000000,38.000000) -- (247.000000,38.000000);
\draw[color=black] (0.000000,38.000000) node[left] {$\ket{0}$};
\draw[color=black] (0.000000,19.000000) -- (247.000000,19.000000);
\draw[color=black] (0.000000,19.000000) node[left] {$\ket{0}$};
\draw[color=black] (0.000000,0.000000) -- (247.000000,0.000000);
\draw[color=black] (0.000000,0.000000) node[left] {$\ket{0}$};
\begin{scope}
\draw[fill=white] (19.000000, 57.000000) +(-45.000000:8.485281pt and 8.485281pt) -- +(45.000000:8.485281pt and 8.485281pt) -- +(135.000000:8.485281pt and 8.485281pt) -- +(225.000000:8.485281pt and 8.485281pt) -- cycle;
\clip (19.000000, 57.000000) +(-45.000000:8.485281pt and 8.485281pt) -- +(45.000000:8.485281pt and 8.485281pt) -- +(135.000000:8.485281pt and 8.485281pt) -- +(225.000000:8.485281pt and 8.485281pt) -- cycle;
\draw (19.000000, 57.000000) node {$H$};
\end{scope}
\begin{scope}
\draw[fill=white] (19.000000, 38.000000) +(-45.000000:8.485281pt and 8.485281pt) -- +(45.000000:8.485281pt and 8.485281pt) -- +(135.000000:8.485281pt and 8.485281pt) -- +(225.000000:8.485281pt and 8.485281pt) -- cycle;
\clip (19.000000, 38.000000) +(-45.000000:8.485281pt and 8.485281pt) -- +(45.000000:8.485281pt and 8.485281pt) -- +(135.000000:8.485281pt and 8.485281pt) -- +(225.000000:8.485281pt and 8.485281pt) -- cycle;
\draw (19.000000, 38.000000) node {$H$};
\end{scope}
\begin{scope}
\draw[fill=white] (19.000000, 19.000000) +(-45.000000:18.384776pt and 10.606602pt) -- +(45.000000:18.384776pt and 10.606602pt) -- +(135.000000:18.384776pt and 10.606602pt) -- +(225.000000:18.384776pt and 10.606602pt) -- cycle;
\clip (19.000000, 19.000000) +(-45.000000:18.384776pt and 10.606602pt) -- +(45.000000:18.384776pt and 10.606602pt) -- +(135.000000:18.384776pt and 10.606602pt) -- +(225.000000:18.384776pt and 10.606602pt) -- cycle;
\draw (19.000000, 19.000000) node {$R_y(\theta)$};
\end{scope}
\begin{scope}
\draw[fill=white] (19.000000, -0.000000) +(-45.000000:18.384776pt and 10.606602pt) -- +(45.000000:18.384776pt and 10.606602pt) -- +(135.000000:18.384776pt and 10.606602pt) -- +(225.000000:18.384776pt and 10.606602pt) -- cycle;
\clip (19.000000, -0.000000) +(-45.000000:18.384776pt and 10.606602pt) -- +(45.000000:18.384776pt and 10.606602pt) -- +(135.000000:18.384776pt and 10.606602pt) -- +(225.000000:18.384776pt and 10.606602pt) -- cycle;
\draw (19.000000, -0.000000) node {$R_y(\theta)$};
\end{scope}
\begin{scope}
\draw[fill=white] (50.000000, 57.000000) +(-45.000000:8.485281pt and 8.485281pt) -- +(45.000000:8.485281pt and 8.485281pt) -- +(135.000000:8.485281pt and 8.485281pt) -- +(225.000000:8.485281pt and 8.485281pt) -- cycle;
\clip (50.000000, 57.000000) +(-45.000000:8.485281pt and 8.485281pt) -- +(45.000000:8.485281pt and 8.485281pt) -- +(135.000000:8.485281pt and 8.485281pt) -- +(225.000000:8.485281pt and 8.485281pt) -- cycle;
\draw (50.000000, 57.000000) node {$P$};
\end{scope}
\begin{scope}
\draw[fill=white] (50.000000, 38.000000) +(-45.000000:8.485281pt and 8.485281pt) -- +(45.000000:8.485281pt and 8.485281pt) -- +(135.000000:8.485281pt and 8.485281pt) -- +(225.000000:8.485281pt and 8.485281pt) -- cycle;
\clip (50.000000, 38.000000) +(-45.000000:8.485281pt and 8.485281pt) -- +(45.000000:8.485281pt and 8.485281pt) -- +(135.000000:8.485281pt and 8.485281pt) -- +(225.000000:8.485281pt and 8.485281pt) -- cycle;
\draw (50.000000, 38.000000) node {$P$};
\end{scope}
\begin{scope}
\draw[fill=white] (50.000000, 19.000000) +(-45.000000:8.485281pt and 8.485281pt) -- +(45.000000:8.485281pt and 8.485281pt) -- +(135.000000:8.485281pt and 8.485281pt) -- +(225.000000:8.485281pt and 8.485281pt) -- cycle;
\clip (50.000000, 19.000000) +(-45.000000:8.485281pt and 8.485281pt) -- +(45.000000:8.485281pt and 8.485281pt) -- +(135.000000:8.485281pt and 8.485281pt) -- +(225.000000:8.485281pt and 8.485281pt) -- cycle;
\draw (50.000000, 19.000000) node {$T$};
\end{scope}
\begin{scope}
\draw[fill=white] (50.000000, -0.000000) +(-45.000000:8.485281pt and 8.485281pt) -- +(45.000000:8.485281pt and 8.485281pt) -- +(135.000000:8.485281pt and 8.485281pt) -- +(225.000000:8.485281pt and 8.485281pt) -- cycle;
\clip (50.000000, -0.000000) +(-45.000000:8.485281pt and 8.485281pt) -- +(45.000000:8.485281pt and 8.485281pt) -- +(135.000000:8.485281pt and 8.485281pt) -- +(225.000000:8.485281pt and 8.485281pt) -- cycle;
\draw (50.000000, -0.000000) node {$T$};
\end{scope}
\draw (71.000000,38.000000) -- (71.000000,0.000000);
\begin{scope}
\draw[fill=white] (71.000000, 0.000000) circle(3.000000pt);
\clip (71.000000, 0.000000) circle(3.000000pt);
\draw (68.000000, 0.000000) -- (74.000000, 0.000000);
\draw (71.000000, -3.000000) -- (71.000000, 3.000000);
\end{scope}
\filldraw (71.000000, 38.000000) circle(1.500000pt);
\draw (89.000000,57.000000) -- (89.000000,0.000000);
\begin{scope}
\draw[fill=white] (89.000000, 57.000000) circle(3.000000pt);
\clip (89.000000, 57.000000) circle(3.000000pt);
\draw (86.000000, 57.000000) -- (92.000000, 57.000000);
\draw (89.000000, 54.000000) -- (89.000000, 60.000000);
\end{scope}
\filldraw (89.000000, 0.000000) circle(1.500000pt);
\draw (107.000000,57.000000) -- (107.000000,19.000000);
\begin{scope}
\draw[fill=white] (107.000000, 19.000000) circle(3.000000pt);
\clip (107.000000, 19.000000) circle(3.000000pt);
\draw (104.000000, 19.000000) -- (110.000000, 19.000000);
\draw (107.000000, 16.000000) -- (107.000000, 22.000000);
\end{scope}
\filldraw (107.000000, 57.000000) circle(1.500000pt);
\draw (125.000000,38.000000) -- (125.000000,19.000000);
\begin{scope}
\draw[fill=white] (125.000000, 38.000000) circle(3.000000pt);
\clip (125.000000, 38.000000) circle(3.000000pt);
\draw (122.000000, 38.000000) -- (128.000000, 38.000000);
\draw (125.000000, 35.000000) -- (125.000000, 41.000000);
\end{scope}
\filldraw (125.000000, 19.000000) circle(1.500000pt);
\draw (143.000000,57.000000) -- (143.000000,38.000000);
\begin{scope}
\draw[fill=white] (143.000000, 57.000000) circle(3.000000pt);
\clip (143.000000, 57.000000) circle(3.000000pt);
\draw (140.000000, 57.000000) -- (146.000000, 57.000000);
\draw (143.000000, 54.000000) -- (143.000000, 60.000000);
\end{scope}
\filldraw (143.000000, 38.000000) circle(1.500000pt);
\begin{scope}
\draw[fill=white] (164.000000, 38.000000) +(-45.000000:8.485281pt and 8.485281pt) -- +(45.000000:8.485281pt and 8.485281pt) -- +(135.000000:8.485281pt and 8.485281pt) -- +(225.000000:8.485281pt and 8.485281pt) -- cycle;
\clip (164.000000, 38.000000) +(-45.000000:8.485281pt and 8.485281pt) -- +(45.000000:8.485281pt and 8.485281pt) -- +(135.000000:8.485281pt and 8.485281pt) -- +(225.000000:8.485281pt and 8.485281pt) -- cycle;
\draw (164.000000, 38.000000) node {$P$};
\end{scope}
\begin{scope}
\draw[fill=white] (164.000000, 19.000000) +(-45.000000:8.485281pt and 8.485281pt) -- +(45.000000:8.485281pt and 8.485281pt) -- +(135.000000:8.485281pt and 8.485281pt) -- +(225.000000:8.485281pt and 8.485281pt) -- cycle;
\clip (164.000000, 19.000000) +(-45.000000:8.485281pt and 8.485281pt) -- +(45.000000:8.485281pt and 8.485281pt) -- +(135.000000:8.485281pt and 8.485281pt) -- +(225.000000:8.485281pt and 8.485281pt) -- cycle;
\draw (164.000000, 19.000000) node {$T$};
\end{scope}
\begin{scope}
\draw[fill=white] (164.000000, -0.000000) +(-45.000000:8.485281pt and 8.485281pt) -- +(45.000000:8.485281pt and 8.485281pt) -- +(135.000000:8.485281pt and 8.485281pt) -- +(225.000000:8.485281pt and 8.485281pt) -- cycle;
\clip (164.000000, -0.000000) +(-45.000000:8.485281pt and 8.485281pt) -- +(45.000000:8.485281pt and 8.485281pt) -- +(135.000000:8.485281pt and 8.485281pt) -- +(225.000000:8.485281pt and 8.485281pt) -- cycle;
\draw (164.000000, -0.000000) node {$Z$};
\end{scope}
\begin{scope}
\draw[fill=white] (199.500000, 57.000000) +(-45.000000:19.798990pt and 10.606602pt) -- +(45.000000:19.798990pt and 10.606602pt) -- +(135.000000:19.798990pt and 10.606602pt) -- +(225.000000:19.798990pt and 10.606602pt) -- cycle;
\clip (199.500000, 57.000000) +(-45.000000:19.798990pt and 10.606602pt) -- +(45.000000:19.798990pt and 10.606602pt) -- +(135.000000:19.798990pt and 10.606602pt) -- +(225.000000:19.798990pt and 10.606602pt) -- cycle;
\draw (199.500000, 57.000000) node {$R_y(\frac{\pi}{4})$};
\end{scope}
\begin{scope}
\draw[fill=white] (199.500000, 38.000000) +(-45.000000:24.748737pt and 10.606602pt) -- +(45.000000:24.748737pt and 10.606602pt) -- +(135.000000:24.748737pt and 10.606602pt) -- +(225.000000:24.748737pt and 10.606602pt) -- cycle;
\clip (199.500000, 38.000000) +(-45.000000:24.748737pt and 10.606602pt) -- +(45.000000:24.748737pt and 10.606602pt) -- +(135.000000:24.748737pt and 10.606602pt) -- +(225.000000:24.748737pt and 10.606602pt) -- cycle;
\draw (199.500000, 38.000000) node {$R_y(-\frac{\pi}{4})$};
\end{scope}
\begin{scope}
\draw[fill=white] (199.500000, 19.000000) +(-45.000000:8.485281pt and 8.485281pt) -- +(45.000000:8.485281pt and 8.485281pt) -- +(135.000000:8.485281pt and 8.485281pt) -- +(225.000000:8.485281pt and 8.485281pt) -- cycle;
\clip (199.500000, 19.000000) +(-45.000000:8.485281pt and 8.485281pt) -- +(45.000000:8.485281pt and 8.485281pt) -- +(135.000000:8.485281pt and 8.485281pt) -- +(225.000000:8.485281pt and 8.485281pt) -- cycle;
\draw (199.500000, 19.000000) node {$H$};
\end{scope}
\begin{scope}
\draw[fill=white] (199.500000, -0.000000) +(-45.000000:24.748737pt and 10.606602pt) -- +(45.000000:24.748737pt and 10.606602pt) -- +(135.000000:24.748737pt and 10.606602pt) -- +(225.000000:24.748737pt and 10.606602pt) -- cycle;
\clip (199.500000, -0.000000) +(-45.000000:24.748737pt and 10.606602pt) -- +(45.000000:24.748737pt and 10.606602pt) -- +(135.000000:24.748737pt and 10.606602pt) -- +(225.000000:24.748737pt and 10.606602pt) -- cycle;
\draw (199.500000, -0.000000) node {$R_y(\frac{3\pi}{4})$};
\end{scope}
\begin{scope}
\draw[fill=white] (235.000000, 57.000000) +(-45.000000:8.485281pt and 8.485281pt) -- +(45.000000:8.485281pt and 8.485281pt) -- +(135.000000:8.485281pt and 8.485281pt) -- +(225.000000:8.485281pt and 8.485281pt) -- cycle;
\clip (235.000000, 57.000000) +(-45.000000:8.485281pt and 8.485281pt) -- +(45.000000:8.485281pt and 8.485281pt) -- +(135.000000:8.485281pt and 8.485281pt) -- +(225.000000:8.485281pt and 8.485281pt) -- cycle;
\draw (235.000000, 57.000000) node {$P$};
\end{scope}
\begin{scope}
\draw[fill=white] (235.000000, 38.000000) +(-45.000000:8.485281pt and 8.485281pt) -- +(45.000000:8.485281pt and 8.485281pt) -- +(135.000000:8.485281pt and 8.485281pt) -- +(225.000000:8.485281pt and 8.485281pt) -- cycle;
\clip (235.000000, 38.000000) +(-45.000000:8.485281pt and 8.485281pt) -- +(45.000000:8.485281pt and 8.485281pt) -- +(135.000000:8.485281pt and 8.485281pt) -- +(225.000000:8.485281pt and 8.485281pt) -- cycle;
\draw (235.000000, 38.000000) node {$P^{\dag}$};
\end{scope}
\begin{scope}
\draw[fill=white] (235.000000, 19.000000) +(-45.000000:8.485281pt and 8.485281pt) -- +(45.000000:8.485281pt and 8.485281pt) -- +(135.000000:8.485281pt and 8.485281pt) -- +(225.000000:8.485281pt and 8.485281pt) -- cycle;
\clip (235.000000, 19.000000) +(-45.000000:8.485281pt and 8.485281pt) -- +(45.000000:8.485281pt and 8.485281pt) -- +(135.000000:8.485281pt and 8.485281pt) -- +(225.000000:8.485281pt and 8.485281pt) -- cycle;
\draw (235.000000, 19.000000) node {$Z$};
\end{scope}
\begin{scope}
\draw[fill=white] (235.000000, -0.000000) +(-45.000000:8.485281pt and 8.485281pt) -- +(45.000000:8.485281pt and 8.485281pt) -- +(135.000000:8.485281pt and 8.485281pt) -- +(225.000000:8.485281pt and 8.485281pt) -- cycle;
\clip (235.000000, -0.000000) +(-45.000000:8.485281pt and 8.485281pt) -- +(45.000000:8.485281pt and 8.485281pt) -- +(135.000000:8.485281pt and 8.485281pt) -- +(225.000000:8.485281pt and 8.485281pt) -- cycle;
\draw (235.000000, -0.000000) node {$P$};
\end{scope}
\end{tikzpicture}

\caption{Quantum circuits generating the states $\ket{L}$, $\ket{\Phi_5}$ and $\ket{M_{2222}}$ up to a global phase. For better readability, most of the rotations around the $\hat y$ and $\hat z$ axes defined by the matrices of parameters are written using the universal single-qubit gates $H,P,T$ (see Identity \eqref{Rz-classical}).\label{circuits}}
 \end{figure}
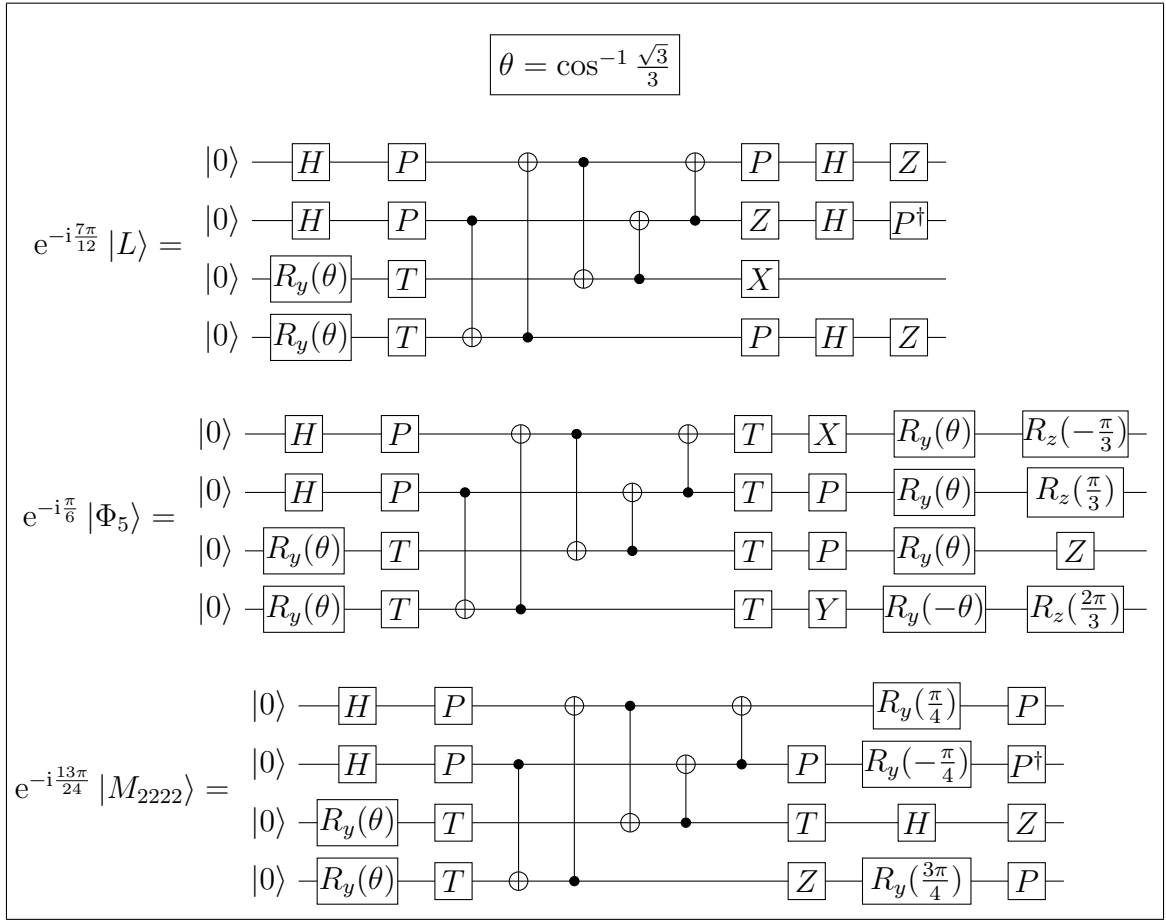

 \medskip
 
 We implemented the circuit generating the state $\ket{L}$ in one of the quantum computers publicly available at \href{https://quantum-computing.ibm.com/}{https://quantum-computing.ibm.com/}. In those computers full connectivity between the qubits is not achieved and the direct connections allowed between two qubits are given by a graph. Moreover, due to the noise in the gates, and particularly the 2-qubit gates, it is of crucial importance to use as few $\cnot$ gates as possible. We chose the 5-qubit
 ibmq\_quito computer because its graph is $\{\{1, 0\}, \{1, 2\}, \{1, 3\}, \{3, 4\}\}$, hence the gates $X_{[01]}, X_{[12]}$ and $X_{[31]}$ of the $\cnot$ subcircuit that implements the operator $M^{(0,1)}_2=X_{[01]}X_{[12]}X_{[20]}X_{[03]}X_{[31]}$ are already native gates. The two other gates of the $\cnot$ subcircuits, namely $X_{[20]}$ and $X_{[03]}$ are not native $\cnot$ gates and can be simulated thanks to the use of $\swap$ gates. Finally the operator $M^{(0,1)}_2$ can be implemented using only 11 native $\cnot$ gates :
     \begin{equation}
M^{(0,1)}_2 = X_{[01]}X_{[12]}\underbrace{X_{[01]}X_{[10]}X_{[01]}}_{S_{(01)}}X_{[21]}X_{[13]}\underbrace{X_{[01]}X_{[10]}X_{[01]}}_{S_{(01)}}X_{[31]}.
\end{equation}
After compilation by the IBM algorithm (the process is called \textit{transpilation} on the website), the quantum circuit implementing the state $\ket{L}$ uses 22 single-qubit native gates, 11 $\cnot$ gates and has a total depth of 18  (see Figure \ref{L-on-IBM}). However, despite of this moderate length, we observed after measurement the apparition of a large quantity of scorias (see the bar chart in Figure \ref{L-on-IBM}). Indeed, from Identity \eqref{L-def}, one has
\begin{equation}
\ket{L}=\textstyle\frac{1}{2}\ee^{\ii\frac{\pi}{6}}(\ket{0000} + \ket{1111}) + \textstyle\frac{\sqrt3}{6}\ee^{-\ii\frac{\pi}{3}}(\ket{0011} + \ket{0101} + \ket{0110} + \ket{1001} + \ket{1010} + \ket{1100}),
  \end{equation}
  so the states   $\ket{0001}, \ket{0010}, \ket{0100}, \ket{1000}, \ket{0111} \ket{1011}, \ket{1101}$ or $\ket{1110}$ should not appear after measurement. The main causes of this problem are, on one hand the measurement errors (average readout error is about $3$ percent on this device), on the other hand the noise in the gates ($\cnot$ gate average error is about $1.3$ percent). This suggests that there are still significant technological challenges to overcome before we can implement the state $\ket{L}$ in a reliable fashion.
  \begin{figure}
    \vspace{2mm}
    
  Original circuit :\vspace{-5mm}
  
  \quad\includegraphics[scale=0.34, viewport=0cm 0.5cm 30.5cm 11cm, clip=true]{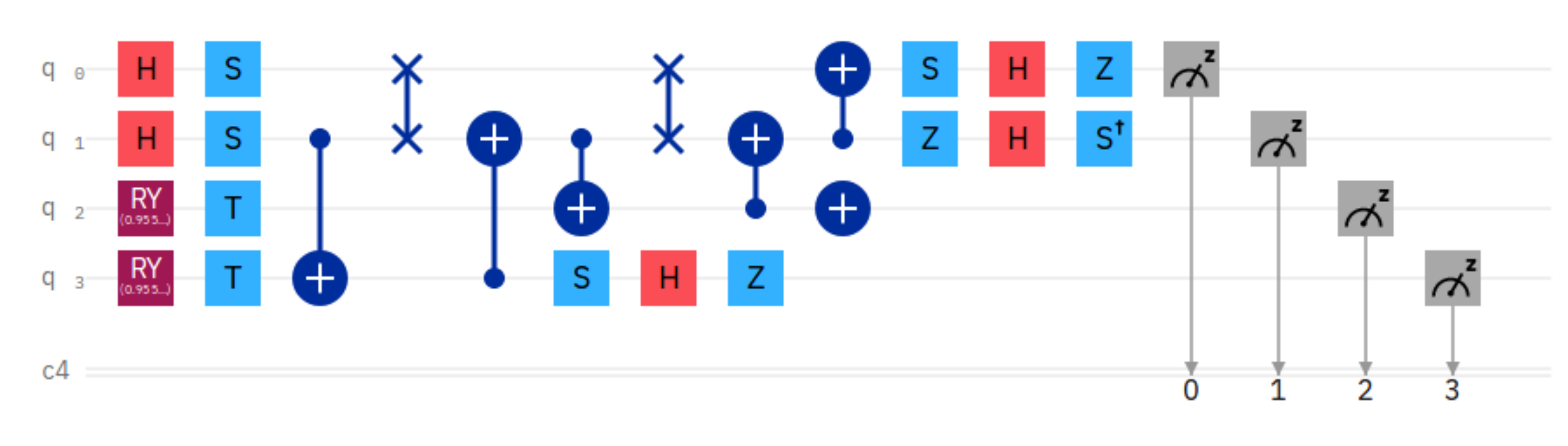}

  Transpiled circuit :
  
  \quad\includegraphics[scale=0.34, viewport=0cm 0.5cm 41cm 11cm, clip=true]{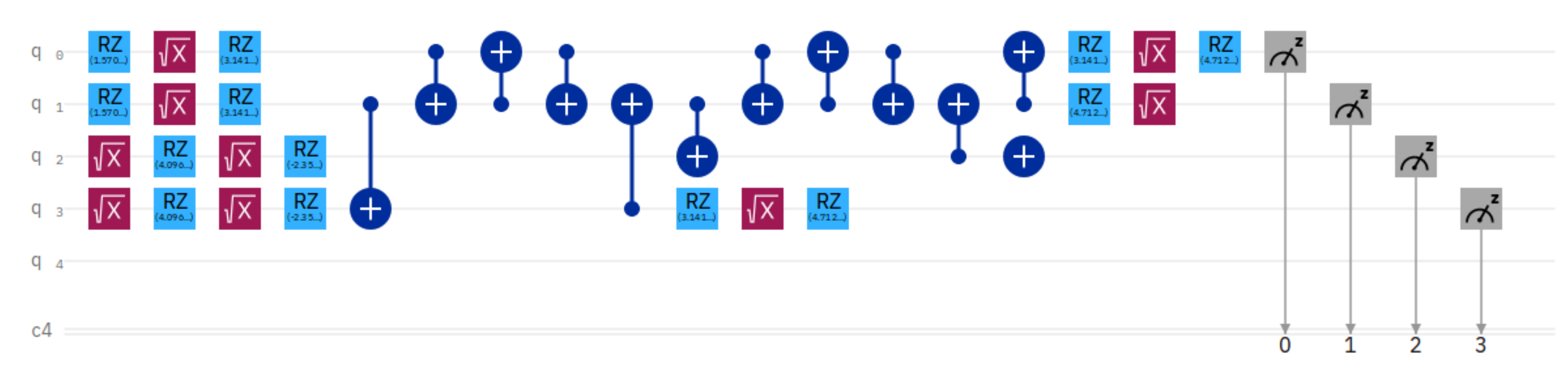}

  Measurements :\vspace{2mm}
  
  \quad\includegraphics[scale=0.4, viewport=0cm 0cm 32.5cm 8cm, clip=true]{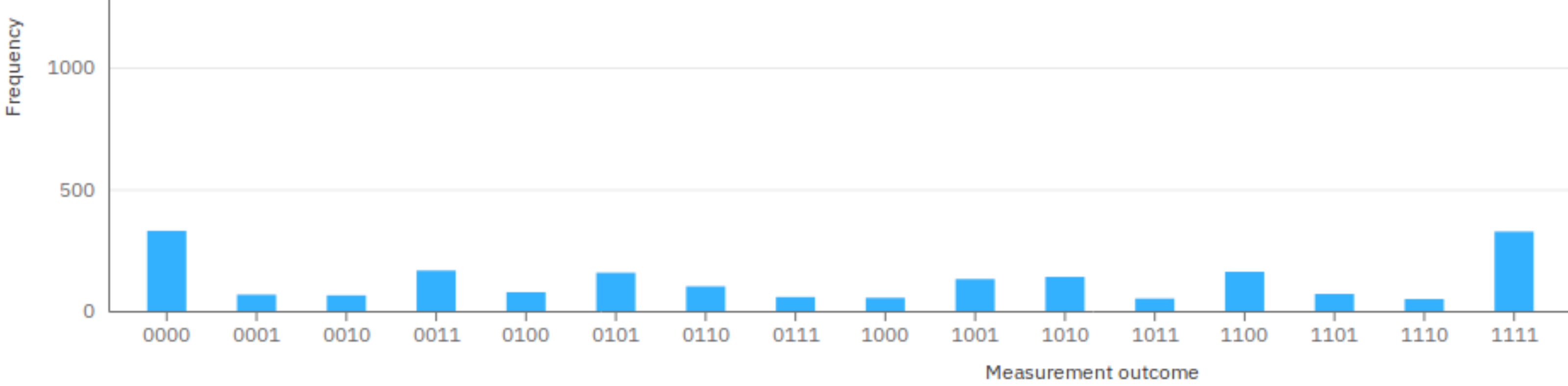}
\caption{Implementation in the ibmq\_quito quantum computer of a circuit generating the state $\ket{L}$. The bar chart is based on 2000 measurements. \label{L-on-IBM}}
\end{figure}
\section{Conclusion and perspectives}
In this work we described how a $\cnot$ circuit acting on a factorized state can produce four-qubit maximum hyperdeterminant states and  we proposed a quantum circuit generating the state $\ket{L}$, whose interesting properties where described by Gour and Wallach \cite{2012GW} and by Chen and Djokovic \cite{2013CD}. It would be interesting to know whether it is possible to generalize this result when the number of qubits $n$ is greater than 4. Is it still possible to reach a MHS by a $\cnot$ circuit acting on a factorized state ? What would be in this case the generalization of the unitary $M^{(i,j)}_k$ to higher dimensions ?
However, answering these questions seems to be currently out of reach because an explicit polynomial expression of the hyperdeterminant is known only up to 4 qubits.
A first approach would be to  know if a generically entangled state (\textit{i.e.} a state $\ket{\psi}$ such that $\Delta_n(\ket{\psi})\neq 0$) can be produced by a $\cnot$ circuit acting on a factorized state in the case of any $n$-qubit system. Indeed, the vanishing of $\Delta_n$ can be tested using the following criterion \cite[p. 445]{1992GKL} :
let $A = \displaystyle\sum_{0\leq i_0,\dots,i_{n-1}\leq 1}a_{i_0\dots i_{n-1}}x^{(0)}_{i_0}\dots x^{(n-1)}_{i_{n-1}}$ be the multilinear form associated to the $n$-qubit state
$\ket{\psi}= \displaystyle\sum_{0\leq i_0,\dots,i_{n-1}\leq 1}a_{i_0\dots i_{n-1}}\ket{i_0\dots i_{n-1}}$, then the condition $\Delta_n(\ket{\psi})= 0$ means that the system
\begin{equation}
  \{A=\frac{\partial A}{\partial x^{(0)}_{0}}=\frac{\partial A}{\partial x^{(0)}_{1}}=\dots=\frac{\partial A}{\partial x^{(n-1)}_{0}}=\frac{\partial A}{\partial x^{(n-1)}_{1}}=0\}
  \end{equation}
  has a solution $(x^{(0)}_{0}, x^{(0)}_{1},\dots, x^{(n-1)}_{0}, x^{(n-1)}_{1})$ such that $( x^{(i)}_{0}, x^{(i)}_{1})\neq (0,0)$ for any $i=0\dots  n-1$. Such a solution is called non trivial. Therefore, to show that a state $\ket{\psi}$ is generically entangled, it is sufficient to prove that the system corresponding to $\ket{\psi}$ has no solutions apart from the trivial solutions. We will go back to these questions in future works.

\section{Acknowledgements}
The author acknowledges the use of the IBM Quantum Experience at \url{https://quantum-computing.ibm.com/}. The views
expressed are those of the author and do not reflect the official policy or position of IBM or
the IBM Quantum Experience team. 

Part of this work was performed using computing resources of the CRIANN (Normandy, France).

\bibliographystyle{plain}
\bibliography{biblio_MB}

\begin{thebibliography}{10}

\bibitem{2017Alsina}
Daniel Alsina.
\newblock Phd thesis: Multipartite entanglement and quantum algorithms, 2017.
\newblock arXiv:1706.08318.

\bibitem{2020B}
Marc Bataille.
\newblock Quantum circuits of {CNOT} gates, 2020.
\newblock arXiv:2009.13247.

\bibitem{2013CD}
Lin Chen and Dragomir~Z. Djokovic.
\newblock Proof of the {Gour-Wallach} conjecture.
\newblock {\em Physical Review A}, 88(4), oct 2013.

\bibitem{1992GKL}
Israel~M {Gelfand}, Mikhail~M. Kapranov, and Zelevisnky~Andrei V.
\newblock {\em Discriminants, Resultants and Multidimensional Determinant}.
\newblock Birkh\"auser, 1992.

\bibitem{2010GW}
Gilad Gour and Nolan~R. Wallach.
\newblock All maximally entangled four-qubit states.
\newblock {\em Journal of Mathematical Physics}, 51(11):112201, Nov 2010.

\bibitem{2012GW}
Gilad Gour and Nolan~R. Wallach.
\newblock On symmetric {SL}-invariant polynomials in four qubits, 2012.

\bibitem{2003LT}
Jean-Gabriel Luque and Jean-Yves Thibon.
\newblock The polynomial invariants of four qubits.
\newblock {\em Phys. Rev. A}, 67:042303, 2003.

\bibitem{2003Miyake}
Akimasa Miyake.
\newblock Classification of multipartite entangled states by multidimensional
  determinant.
\newblock {\em Phys. Rev. A}, 67:012108, 2003.

\bibitem{2002MW}
Akimasa Miyake and Miki Wadati.
\newblock Multipartite entanglement and hyperdeterminants.
\newblock {\em Quantum Information and Computation}, 2:540--555, 2002.

\bibitem{2011NC}
Michael~A. Nielsen and Isaac~L. Chuang.
\newblock {\em Quantum Computation and Quantum Information: 10th Anniversary
  Edition}.
\newblock Cambridge University Press, New York, NY, USA, 10th edition, 2011.

\bibitem{2006OS}
Andreas Osterloh and Jens Siewert.
\newblock Entanglement monotones and maximally entangled states in multipartite
  qubit systems.
\newblock {\em International Journal of Quantum Information}, 04(03):531–540,
  Jun 2006.

\end{thebibliography}

\end{document}